\documentclass[12pt,onecolumn,draftclsnofoot]{IEEEtran}

\usepackage{amsmath,amssymb,amsthm}
\usepackage{cite}
\usepackage[hidelinks]{hyperref}
\usepackage{graphicx}

\newtheorem{theorem}{Theorem}
\newtheorem{proposition}{Proposition}
\newtheorem{corollary}{Corollary}

\newcommand{\R}{\mathbb{R}}
\newcommand{\Sph}{\mathbb{S}}
\newcommand{\E}{\mathbb{E}}
\newcommand{\Prob}{\mathbb{P}}
\newcommand{\dd}{\mathrm{d}}
\newcommand{\Pe}{P_{\mathrm e}}

\title{Random Spherical Codes at High SNR: Error Transitions, Fixed-Error Data Rates, and Converse Gaps}

\author{Nikola Zlatanov
\thanks{N.~Zlatanov is with Innopolis University, Innopolis, 420500, Russia (e-mail: n.zlatanov@innopolis.ru).}
}

\begin{document}
\maketitle

\begin{abstract}
This paper characterizes random spherical codebooks over the real
additive white Gaussian noise channel in the high signal-to-noise ratio
(SNR) regime in which the blocklength is fixed, the SNR per real
channel use tends to infinity, and the codebook size grows with SNR.
The ensemble exhibits a sharp error-probability transition governed by
the intrinsic dimension of the sphere and the codebook-growth order.  Below the critical
codebook-growth order, the ensemble-average error probability vanishes;
at the critical order, it converges to a nontrivial limit; and above
that order, it approaches one.  For a fixed target average error
probability, this transition yields the high-SNR expansion of the
ensemble-achievable data rate.  The achievable rate and the
corresponding converse rate bound have the same high-SNR prelog,
establishing first-order optimality within the deterministic
equal-energy, average-error class.  Their ratio tends to one for every
fixed blocklength and target error probability, but their additive
difference approaches a strictly positive limit that depends on both.
We characterize this rate-bound gap jointly in blocklength and error
probability.  At fixed error probability, it vanishes with increasing
blocklength, with a universal leading $1/n$ behavior and reliability
dependence first appearing at the next order.  For error probabilities
that decrease exponentially with blocklength, we identify the threshold
between vanishing and nonvanishing limiting gaps.  We also derive
blocklength laws for meeting a prescribed gap as the target error
probability becomes more stringent.
\end{abstract}

\begin{IEEEkeywords}
AWGN channels, finite blocklength, fixed-error data rate, high-SNR
asymptotics, random coding, spherical codes.
\end{IEEEkeywords}

\section{Introduction}
\label{sec:introduction}

Random spherical codebooks provide a natural meeting point between
coding geometry and finite-blocklength analysis of the additive white
Gaussian noise (AWGN) channel.  Every codeword has exactly the same
energy, while rotational invariance makes the geometry of
maximum-likelihood (ML) decoding analytically accessible.  This paper
studies these codebooks in the fixed-blocklength, high signal-to-noise ratio (SNR) regime: the
blocklength $n$ is fixed, the SNR per real channel use $\gamma$ tends
to infinity, and the codebook size $M=M(\gamma)$ is allowed to grow
with SNR.  With $n$ fixed, a fixed rate corresponds to a fixed codebook
size and hence an error probability that vanishes at high SNR.  A
nondegenerate fixed-error limit can arise only if the codebook itself
grows with SNR.  The central questions are therefore how rapidly
$M(\gamma)$ may grow, what fixed-error data rate this growth supports, and
how closely that rate approaches the upper bound supplied by an
equal-energy converse.

Classical AWGN coding theory establishes capacity and error-exponent
laws and develops geometric sphere-packing and circular-cone converses
\cite{ClassicalInfo1948,GaussianOptimalCodes1959,
ReliableCommunication1968}.  Modern finite-blocklength theory provides
nonasymptotic achievability and converse bounds together with
large-blocklength normal approximations
\cite{FiniteBlocklength2010}.  The circular-cone converse has also been
identified with the optimized minimax converse for the equal-energy
AWGN setting \cite{PolyanskiyMinimax2013}, while direct numerical
methods, Laplace-integral techniques, and refined tail-probability
bounds have been used to evaluate finite-blocklength AWGN converses or
approximate them asymptotically \cite{AWGNConverseEval2015,
FiniteBlocklengthLaplace2016,TailProbabilityBounds2026}.  In parallel,
exact and computationally efficient evaluations of ensemble-average
error probability have been developed for random spherical and related
code ensembles \cite{RandomCodeEnsembleError2021,
EfficientRandomCodingError2023}.  These lines of work do not directly
characterize the asymptotic axis considered here, in which $n$ remains
fixed while both $\gamma$ and $M(\gamma)$ grow.  The present paper
connects the geometric, random-coding, and converse viewpoints on this
axis and carries both the explicit achievable and converse bounds
through their constant high-SNR terms.

Our derivation follows the high-SNR geometry of ML decoding for
spherical codebooks.  Once the transmitted codeword and Gaussian noise
are fixed, the received vector selects a spherical cap containing
precisely the locations of the competitor codewords that tie or beat
the transmitted codeword in the ML metric.  Because the $M-1$
competitor codewords are drawn independently, they enter this
\emph{dangerous cap} independently, conditional on the transmitted
codeword and noise.  Hence, the ensemble-average error probability is
determined exactly by the cap probability and the conditional
probability that none of the competitor codewords enters the cap.  We
express the cap measure in regularized-beta form and then show that,
after multiplication by $\gamma^{(n-1)/2}$, it converges to an explicit
positive random limit.  Substituting this limit into the exact
probability that no competitor codeword enters the cap yields the three
error-probability regimes described below.  We then invert the
nondegenerate transition curve to obtain the fixed-error data rate, expand
the rate bound induced by the circular-cone converse on the same
high-SNR scale, and compare them
to characterize the rate-bound gap and the
associated blocklength--reliability tradeoffs.

This geometry identifies the scaled codebook size
\begin{equation}
  \mu_\gamma
  =
  \frac{M(\gamma)-1}{\gamma^{(n-1)/2}}.
  \label{eq:intro-scaled-codebook-size}
\end{equation}
For a growing codebook $M(\gamma)\to\infty$, the ensemble-average error
probability obeys the sharp three-regime transition
\begin{subequations}
\label{eq:intro-error-transition}
\begin{equation}
  \begin{gathered}
    \mu_\gamma\to0 \quad\Longrightarrow \quad
    \overline\Pe(M(\gamma),n;\gamma)
    \sim \E[\Theta_n] \mu_\gamma\to0.
  \end{gathered}
  \label{eq:intro-underloaded}
\end{equation}
\begin{equation}
  \begin{aligned}
    \mu_\gamma\to\mu\in(0,\infty)
     \quad\Longrightarrow 
    \overline\Pe(M(\gamma),n;\gamma)
    &\to1-\E[e^{-\mu\Theta_n}]\in(0,1).
  \end{aligned}
  \label{eq:intro-critical}
\end{equation}
\begin{equation}
  \mu_\gamma\to\infty
  \quad\Longrightarrow\quad
  \overline\Pe(M(\gamma),n;\gamma)
  \to1,\qquad \qquad\;\;
  \label{eq:intro-overloaded}
\end{equation}
\end{subequations}
where $\Theta_n$ is the
limiting random variable for the scaled dangerous-cap probability and
captures the transverse-noise-induced fluctuation.  The three cases in
\eqref{eq:intro-error-transition} are called the underloaded,
critically loaded, and overloaded regimes.
In the underloaded regime, the number $M-1$ of independent competitor
codewords grows too slowly to offset the $\gamma^{-(n-1)/2}$ decay of
the probability that a given competitor codeword enters the dangerous
cap.  At
critical loading, these two effects balance and
the noise-induced fluctuation $\Theta_n$ remains visible in the limiting
error probability.  In the overloaded regime, the growth in the number
of competitor codewords dominates the decay of the cap-entry
probability of a single competitor codeword.  Thus
$M(\gamma)\asymp\gamma^{(n-1)/2}$ is the only codebook-size order
capable of producing a nondegenerate limiting error probability.

The critical transition of the ensemble-average error probability can
be inverted to obtain the data rate under a prescribed error probability
$\varepsilon$.  Specifically, let
$R_{\varepsilon,\mathrm{ens}}(n,\gamma)$ denote the data rate associated
with the largest codebook size whose ensemble-average error probability
does not exceed $\varepsilon$.  Then
\begin{align}
  R_{\varepsilon,\mathrm{ens}}(n,\gamma)
  &=
  \frac{n-1}{2n}\log_2\gamma
  +\frac1n\log_2
  \mu_{n,\varepsilon}^{\mathrm{ens}}
  \nonumber\\
  &\quad
  +O_{n,\varepsilon}(\gamma^{-1/2}),
  \label{eq:intro-fixed-error-rate}
\end{align}
where, for a target error probability $0<\varepsilon<1$,
$\mu_{n,\varepsilon}^{\mathrm{ens}}$ is the unique positive solution of
\begin{equation}
  1-\E\!\left[
  e^{-\mu_{n,\varepsilon}^{\mathrm{ens}}\Theta_n}
  \right]
  =\varepsilon.
  \label{eq:intro-fixed-error-root}
\end{equation}
This expansion separates the leading geometric growth from the
constant-order reliability dependence.  The prelog $(n-1)/(2n)$ is
determined by the $n-1$ local dimensions available on the sphere,
whereas, for fixed $n$,
$n^{-1}\log_2\mu_{n,\varepsilon}^{\mathrm{ens}}$ contains all
dependence on the prescribed error probability $\varepsilon$ at constant order.
Ensemble averaging guarantees the existence, at each SNR, of at least
one deterministic spherical codebook of that size whose average error
probability does not exceed $\varepsilon$.

To assess first-order optimality, let
$R_{\varepsilon,\mathrm{conv}}(n,\gamma)$ denote the fixed-error rate
upper bound induced by the circular-cone converse for deterministic
equal-energy codes under average error probability.  We show that
\begin{equation}
  1-
  \frac{R_{\varepsilon,\mathrm{ens}}(n,\gamma)}
       {R_{\varepsilon,\mathrm{conv}}(n,\gamma)}
  =
  O_{n,\varepsilon}\!\left(\frac1{\ln\gamma}\right).
  \label{eq:intro-relative-optimality}
\end{equation}
Ensemble averaging guarantees a deterministic spherical code attaining
$R_{\varepsilon,\mathrm{ens}}$, whereas the converse upper bounds every
deterministic equal-energy code under average error probability by
$R_{\varepsilon,\mathrm{conv}}$.  Their common leading term therefore
establishes the optimal high-SNR prelog $(n-1)/(2n)$ within this class,
and the relative shortfall between the two explicit bounds vanishes as
$O(1/\ln\gamma)$.  This is a first-order statement: at fixed $n$, their
additive difference generally approaches a positive constant.

The remaining difference between the explicit achievable data rate and
the converse bound is captured by the high-SNR rate-bound gap
$G(n,\varepsilon)$.  It is defined as the limiting difference
$R_{\varepsilon,\mathrm{conv}}-R_{\varepsilon,\mathrm{ens}}$ as
$\gamma\to\infty$; thus the finite-SNR rate difference approaches a
floor.  For every
fixed $\varepsilon$, as $n\to\infty$, the gap is
$G(n,\varepsilon)=\gamma_{\mathrm E}/(n\ln2)
+O_\varepsilon(n^{-3/2})$, where $\gamma_{\mathrm E}$ is the
Euler--Mascheroni constant.  Dependence on reliability first appears beyond the leading
$1/n$ term.
At fixed $n$ and increasingly stringent reliability, replacing $\varepsilon$ by
$\varepsilon/10$, i.e., a tenfold reduction in the target error
probability, increases the gap asymptotically by
$\log_2(10)/n$ bits per real channel use.  When blocklength and
reliability scale jointly according to
$\varepsilon_n=e^{-\rho n+o(n)}$, the gap converges to zero as
$n\to\infty$ for $0<\rho\le(1-\ln2)/2$; for
$\rho>(1-\ln2)/2$, it converges to a positive constant.  We provide
explicit blocklength laws for meeting a prescribed gap as the permitted
gap decreases or the reliability requirement becomes more stringent.

The remainder of the paper is organized as follows.
Section~\ref{sec:model} introduces the spherical-code model and derives
the exact dangerous-cap representation.  Section~\ref{sec:load-scaled}
establishes the growing-codebook error transition.
Section~\ref{sec:fixed-error-rates} obtains the data rate at a
prescribed error probability.
Section~\ref{sec:deterministic-converse-comparison} expands the known
circular-cone converse on the same high-SNR scale and establishes
first-order optimality.
Section~\ref{sec:gap-design} analyzes the rate-bound gap and the joint
blocklength--reliability tradeoffs.
Section~\ref{sec:numerical-results} presents numerical results, and
Section~\ref{sec:conclusion} concludes the paper.

\section{Model and Exact Ensemble-Average Error Representation}
\label{sec:model}

This section defines the equal-energy AWGN model and reduces the
ensemble-average ML error probability to the normalized surface measure
of one noise-selected spherical cap.  The representation is exact for
every finite $(M,n,\gamma)$ and provides the starting point for the
high-SNR analysis.

\subsection{Channel, spherical ensemble, and ML decoding}

Let
\[
  \Sph^{n-1}=\{\mathbf x\in\R^n:\|\mathbf x\|=1\}
\]
and let
\[
  \mathcal C=\{\mathbf c_1,\ldots,\mathbf c_M\}
  \subset\Sph^{n-1}
\]
be an $(M,n)$ equal-energy codebook.  Its rate is
\[
  \frac1n\log_2M
  \qquad\text{bits per real channel use}.
\]
For $[M]=\{1,\ldots,M\}$ and a uniformly distributed message
$W\in[M]$, the channel output is
\begin{equation}
  \mathbf Y=\sqrt{n\gamma}\,\mathbf C_W+\mathbf Z,
  \qquad
  \mathbf Z\sim\mathcal N(\mathbf0,\mathbf I_n),
  \label{eq:channel}
\end{equation}
where, under the unit-noise normalization, $\gamma$ is the received SNR
per real channel use.  Indeed, every transmitted vector satisfies
\begin{equation}
  \frac1n\left\|\sqrt{n\gamma}\,\mathbf C_W\right\|^2
  =\gamma.
  \label{eq:per-use-snr-normalization}
\end{equation}
Throughout the principal asymptotic
analysis, $n\ge2$ is fixed and $\gamma\to\infty$.  We write
\begin{equation}
  a=\frac{n-1}{2}.
  \label{eq:a-def}
\end{equation}
The exponent $a$ is half the dimension of the tangent space of
$\Sph^{n-1}$ and will determine the high-SNR scaling.

Since all codewords have equal norm, ML decoding is nearest-neighbor
decoding.  The random spherical ensemble is
\begin{equation}
  \mathbf C_1,\ldots,\mathbf C_M
  \stackrel{\mathrm{i.i.d.}}{\sim}
  \operatorname{Unif}(\Sph^{n-1}),
  \label{eq:spherical-ensemble}
\end{equation}
independently of $(W,\mathbf Z)$.  For a deterministic codebook
$\mathcal C$, let
\[
  \Pe(\mathcal C;\gamma)
  =\Prob\{\widehat W\ne W\mid\mathcal C\}
\]
where $\widehat W$ denotes the ML decoder output.  This is the average
error probability over the equiprobable messages.
The ensemble-average error probability is
\begin{equation}
  \overline\Pe(M,n;\gamma)
  =\E_{\mathcal C}[\Pe(\mathcal C;\gamma)].
  \label{eq:avg-pe}
\end{equation}
We use the convention $\overline\Pe(1,n;\gamma)=0$.

\subsection{Dangerous-cap probability}

Conditioning on the transmitted codeword $\mathbf c$ and noise realization
$\mathbf z$, the channel output is $\mathbf y=\sqrt{n\gamma}\,\mathbf c+\mathbf z$.  An
independently drawn competitor codeword $\widetilde{\mathbf C}$ ties or beats
$\mathbf c$ in the ML metric precisely when it belongs to
\begin{equation}
  \mathcal D_\gamma(\mathbf c,\mathbf z)
  =
  \left\{
  \widetilde{\mathbf c}\in\Sph^{n-1}:
  \mathbf y^T\widetilde{\mathbf c}
  \ge \mathbf y^T\mathbf c
  \right\}.
  \label{eq:sph-dangerous-cap}
\end{equation}
This set is the \emph{dangerous cap}.  Let
$\sigma_{n-1}$ denote surface measure on $\Sph^{n-1}$ and define
\begin{align}
  \pi_\gamma(\mathbf c,\mathbf z)
  &=
  \frac{\sigma_{n-1}(\mathcal D_\gamma(\mathbf c,\mathbf z))}
       {\sigma_{n-1}(\Sph^{n-1})}
  \nonumber\\
  &=
  \Prob\{\widetilde{\mathbf C}\in
  \mathcal D_\gamma(\mathbf c,\mathbf z)\}.
  \label{eq:pi-gamma-def}
\end{align}
Thus $\pi_\gamma(\mathbf c,\mathbf z)$ is both the normalized surface
measure of the dangerous cap and the conditional probability that a
uniformly distributed spherical competitor codeword falls in it.  Let
\[
  \Pi_\gamma
  =\pi_\gamma(\mathbf C_W,\mathbf Z)
\]
be the corresponding random dangerous-cap probability.  In
$\Pi_\gamma$, the competitor-codeword location has already been averaged out;
the remaining randomness comes from the transmitted codeword and noise.

Decompose the noise as
\begin{equation}
  Z_\parallel=\mathbf C_W^T\mathbf Z,
  \qquad
  \mathbf Z_\perp=\mathbf Z-Z_\parallel\mathbf C_W,
  \qquad
  T_\perp=\frac{\|\mathbf Z_\perp\|^2}{2}.
  \label{eq:noise-decomposition-main}
\end{equation}
Then $Z_\parallel\sim\mathcal N(0,1)$ and
$T_\perp\sim\operatorname{Gamma}(a,1)$ are independent, where the gamma
distribution uses the shape--rate convention.  The Euclidean height of
the random dangerous cap
$\mathcal D_\gamma(\mathbf C_W,\mathbf Z)$ is
\begin{equation}
  H_\gamma
  =
  1-\frac{\sqrt{n\gamma}+Z_\parallel}
  {\sqrt{(\sqrt{n\gamma}+Z_\parallel)^2+2T_\perp}}.
  \label{eq:Hgamma-def}
\end{equation}
For $0\le x\le1$, let
\begin{equation}
  I_x(a,a)
  =\frac1{\mathrm B(a,a)}
  \int_0^x u^{a-1}(1-u)^{a-1}\,\dd u
  \label{eq:regularized-beta-def}
\end{equation}
where $\mathrm B$ is the beta function; thus $I_x(a,a)$ is the
regularized beta function.

Let $N_\gamma$ be the number of competitor codewords in the random
dangerous cap:
\begin{equation}
  N_\gamma
  =
  \#\left\{i\in[M]\setminus\{W\}:
  \mathbf C_i\in
  \mathcal D_\gamma(\mathbf C_W,\mathbf Z)
  \right\}.
  \label{eq:dangerous-count}
\end{equation}

\begin{proposition}[Exact ensemble-average error formula]
\label{prop:dangerous}
For every $M\ge1$, $n\ge2$, and $\gamma>0$, away from the null event
$\mathbf Y=\mathbf0$,
\begin{equation}
  \Pi_\gamma=I_{H_\gamma/2}(a,a).
  \label{eq:Pi-beta-exact}
\end{equation}
Conditional on $(W,\mathbf C_W,\mathbf Z)$,
\begin{equation}
  N_\gamma
  \sim\operatorname{Binomial}
  \left(M-1,\pi_\gamma(\mathbf C_W,\mathbf Z)\right).
  \label{eq:binomial-dangerous-count}
\end{equation}
Consequently,
\begin{align}
  \overline\Pe(M,n;\gamma)
  &=1-\E[(1-\Pi_\gamma)^{M-1}]
  \label{eq:dangerous-identity}\\
  &=1-\E\left[
  \bigl(1-I_{H_\gamma/2}(a,a)\bigr)^{M-1}
  \right].
  \label{eq:scalar-error-identity}
\end{align}
\end{proposition}

\begin{proof}
For a fixed unit vector $\mathbf v$ and
$\widetilde{\mathbf C}\sim\operatorname{Unif}(\Sph^{n-1})$,
$(1-\widetilde{\mathbf C}^T\mathbf v)/2$ has the
$\operatorname{Beta}(a,a)$ distribution.  Taking
$\mathbf v=\mathbf Y/\|\mathbf Y\|$ converts
\eqref{eq:sph-dangerous-cap} into the event
$(1-\widetilde{\mathbf C}^T\mathbf v)/2\le H_\gamma/2$, proving
\eqref{eq:Pi-beta-exact}.  Given the transmitted codeword and noise, the
$M-1$ competitor codewords independently enter the same dangerous cap with
probability $\pi_\gamma(\mathbf C_W,\mathbf Z)$.  This proves the
conditional binomial law.  An ML error occurs exactly when
$N_\gamma\ge1$, up to a null tie event, and averaging the conditional
zero-count probability gives the two error identities.
\end{proof}

The representation separates the number of competitor codewords from
the single-codeword geometry.  The codebook supplies $M-1$ independent
competitor codewords, whereas $\Pi_\gamma$ is the noise-selected
conditional probability that any one of them defeats the transmitted
codeword.  An error occurs when at least one competitor codeword
succeeds.  The next section identifies the unique codebook-size scale
on which the growth in the number of competitor codewords balances the
decay of $\Pi_\gamma$.

\section{High-SNR Error Transition}
\label{sec:load-scaled}

The error transition follows in two steps.  We first identify the
high-SNR scale of the probability that one competitor codeword enters the
dangerous cap.  We then combine this limit with the exact zero-count
identity in \eqref{eq:dangerous-identity} to determine how the error
probability changes with codebook size.

Let the codebook size depend on SNR.  Define its scaled size by
\begin{equation}
  \mu(M,\gamma)
  =\frac{M-1}{\gamma^{(n-1)/2}},
  \qquad
  \mu_\gamma=\mu(M(\gamma),\gamma).
  \label{eq:normalized-size-def}
\end{equation}

Define
\begin{equation}
  \begin{aligned}
    D_n
    &=
    \frac1{(2n)^a a\,\mathrm B(a,a)}
    \\
    &=
    \left(\frac{2}{n}\right)^{(n-1)/2}
    \frac{\Gamma(n/2)}
    {(n-1)\sqrt\pi\,\Gamma((n-1)/2)}.
  \end{aligned}
  \label{eq:Dn-def}
\end{equation}
Define the \emph{limiting scaled dangerous-cap probability} on the
same noise probability space as $\Pi_\gamma$ by
\begin{equation}
  \Theta_n=D_nT_\perp^a.
  \label{eq:Theta-law}
\end{equation}
Equivalently,
$\Theta_n\stackrel{\mathrm d}=D_nT^a$ for
$T\sim\operatorname{Gamma}(a,1)$.
Its mean is
\begin{equation}
  A_n=\E[\Theta_n]
  =D_n\frac{\Gamma(2a)}{\Gamma(a)}.
  \label{eq:An-def}
\end{equation}
Thus $\Theta_n$ records the cap-size fluctuation produced by the
transverse noise.

\begin{proposition}[Scaled dangerous-cap limit]
\label{prop:local-mass-limit}
Fix $n\ge2$, and write
$\|X\|_q=(\E|X|^q)^{1/q}$.  As $\gamma\to\infty$, for
$q\in\{1,2\}$,
\begin{equation}
  \left\|\gamma^a\Pi_\gamma-\Theta_n\right\|_q
  =O_{n,q}(\gamma^{-1/2}).
  \label{eq:local-mass-Lq}
\end{equation}
In particular,
\begin{equation}
  \E[\Pi_\gamma]
  =\gamma^{-a}\left[A_n+O_n(\gamma^{-1/2})\right].
  \label{eq:p1-leading}
\end{equation}
\end{proposition}

\begin{proof}
See Appendix~\ref{app:local-mass-limits}.
\end{proof}

Proposition~\ref{prop:local-mass-limit} establishes
$\gamma^a\Pi_\gamma\to\Theta_n=D_nT_\perp^a$ in both $L^1$ and
$L^2$, which is the central probabilistic step behind the error and
rate results.  At high SNR, the dangerous-cap half-angle scales as
$\|\mathbf Z_\perp\|/\sqrt{n\gamma}$, so its normalized surface measure
scales as $\gamma^{-a}$; $\Theta_n$ retains the transverse-noise
fluctuation, while longitudinal noise enters only at lower order.

Define the limiting ensemble-average error curve
\begin{align}
  F_n^{\mathrm{ens}}(\mu)
  &=1-\E[e^{-\mu\Theta_n}]
  \nonumber\\
  &=1-\frac1{\Gamma(a)}
  \int_0^\infty t^{a-1}
  e^{-t-\mu D_nt^a}\,\dd t,
  \qquad \mu\ge0.
  \label{eq:ensemble-error-curve}
\end{align}

We now provide the high-SNR transition of the ensemble-average error probability.

\begin{theorem}[Quantitative high-SNR error transition]
\label{thm:load-scaled}
Fix $n\ge2$.  For every finite $K>0$, there are constants
$C_{n,K}<\infty$ and $\gamma_{n,K}<\infty$ such that, for
$\gamma\ge\gamma_{n,K}$,
\begin{equation}
  \sup_{\substack{M\in\mathbb N:\\ \mu(M,\gamma)\le K}}
  \left|
  \overline\Pe(M,n;\gamma)
  -F_n^{\mathrm{ens}}(\mu(M,\gamma))
  \right|
  \le\frac{C_{n,K}}{\sqrt\gamma}.
  \label{eq:quantitative-error-curve}
\end{equation}
Consequently, if $\mu_\gamma\to\mu\in[0,\infty)$, then
\begin{equation}
  \overline\Pe(M(\gamma),n;\gamma)
  \longrightarrow F_n^{\mathrm{ens}}(\mu).
  \label{eq:load-error-limit}
\end{equation}
\end{theorem}

\begin{proof}
See Appendix~\ref{app:first-order-load-proof}.
\end{proof}

\begin{corollary}[Growing-codebook error regimes]
\label{cor:growing-codebook-regimes}
Fix $n\ge2$.  As $\gamma\to\infty$, let
$M=M(\gamma)\to\infty$.  Then
\begin{align}
  \mu_\gamma\to0
  &\quad\Longrightarrow\quad
  \nonumber\\[-2pt]
  &\quad\overline\Pe(M(\gamma),n;\gamma)
  \sim A_n\mu_\gamma\to0,
  \label{eq:underload-regime}\\
  \mu_\gamma\to\mu\in(0,\infty)
  &\quad\Longrightarrow\quad
  \nonumber\\[-2pt]
  &\quad\overline\Pe(M(\gamma),n;\gamma)
  \to F_n^{\mathrm{ens}}(\mu)\in(0,1),
  \label{eq:critical-regime}\\
  \mu_\gamma\to\infty
  &\quad\Longrightarrow\quad
  \nonumber\\[-2pt]
  &\quad\overline\Pe(M(\gamma),n;\gamma)
  \to1.
  \label{eq:overload-regime}
\end{align}
These three cases are called the underloaded, critically loaded, and
overloaded regimes, respectively.  Hence
$M(\gamma)\asymp\gamma^{(n-1)/2}$ is the only codebook-size scale on
which a nondegenerate limiting error probability can occur; when
$\mu_\gamma\to\mu\in(0,\infty)$, the limit is
$F_n^{\mathrm{ens}}(\mu)\in(0,1)$.
\end{corollary}

\begin{proof}
See Appendix~\ref{app:first-order-load-proof}.
\end{proof}

The curve $F_n^{\mathrm{ens}}$ is continuous and strictly increasing,
with $F_n^{\mathrm{ens}}(0)=0$ and
$F_n^{\mathrm{ens}}(\mu)\to1$ as $\mu\to\infty$.  Therefore every
target error probability in $(0,1)$ selects a unique point on the
critical curve.  Thus the scaled codebook size $\mu$ is the sole
first-order parameter that must be inverted to impose a target error.

All limits in this section hold for each fixed integer $n\ge2$; no
uniformity in blocklength is asserted.  Moreover, $\overline\Pe$ is an
ensemble average over the random codebook, message, and noise.

\section{Fixed-Error Data Rate}
\label{sec:fixed-error-rates}

We now invert the critical error curve to obtain the data rate under a
prescribed average error probability.  For fixed $0<\varepsilon<1$,
define
$\mu_{n,\varepsilon}^{\mathrm{ens}}$ as the unique solution of
\begin{equation}
  F_n^{\mathrm{ens}}
  (\mu_{n,\varepsilon}^{\mathrm{ens}})
  =\varepsilon,
  \label{eq:mu-eps-ens-def}
\end{equation}
where $F_n^{\mathrm{ens}}$ is defined in
\eqref{eq:ensemble-error-curve}.
Thus $\mu_{n,\varepsilon}^{\mathrm{ens}}$ is the normalized codebook
size that produces limiting error probability $\varepsilon$ on the
critical scale.
Let
\begin{align}
  M_{\varepsilon,\mathrm{ens}}(n,\gamma)
  &=\max\{M\in\mathbb N:
  \overline\Pe(M,n;\gamma)\le\varepsilon\},
  \label{eq:M-eps-ens-def}\\
  R_{\varepsilon,\mathrm{ens}}(n,\gamma)
  &=\frac1n\log_2
  M_{\varepsilon,\mathrm{ens}}(n,\gamma).
  \label{eq:R-eps-ens-def}
\end{align}
Since $0<\Pi_\gamma<1$ almost surely, the exact representation in
\eqref{eq:dangerous-identity} shows that
$\overline\Pe(M,n;\gamma)$ is strictly increasing in $M$ and converges
to one as $M\to\infty$.  Hence the threshold
$\overline\Pe(M,n;\gamma)\le\varepsilon$ is finite and well defined.

We refer to $R_{\varepsilon,\mathrm{ens}}(n,\gamma)$ as the
ensemble-achievable data rate because $\overline\Pe\le\varepsilon$
guarantees, for each SNR, the existence of at least one deterministic
spherical codebook of that size whose average error probability does
not exceed $\varepsilon$.  The following theorem gives its high-SNR
expansion.

\begin{theorem}[Fixed-error data rate]
\label{thm:ensemble-fixed-error-rate}
For every fixed $n\ge2$ and $0<\varepsilon<1$, as $\gamma\to\infty$,
\begin{align}
  \frac{M_{\varepsilon,\mathrm{ens}}(n,\gamma)-1}
       {\gamma^{(n-1)/2}}
  =\mu_{n,\varepsilon}^{\mathrm{ens}}
  +O_{n,\varepsilon}(\gamma^{-1/2}).
  \label{eq:M-eps-leading}
\end{align}
Thus, with $a=(n-1)/2$,
\begin{equation}
  M_{\varepsilon,\mathrm{ens}}(n,\gamma)
  =1+\mu_{n,\varepsilon}^{\mathrm{ens}}\gamma^a
  +O_{n,\varepsilon}(\gamma^{a-1/2}).
  \label{eq:M-eps-leading-unscaled}
\end{equation}
The corresponding data rate satisfies
\begin{align}
  R_{\varepsilon,\mathrm{ens}}(n,\gamma)
  &=\frac{n-1}{2n}\log_2\gamma
  +\frac1n\log_2
  \mu_{n,\varepsilon}^{\mathrm{ens}}
  \nonumber\\
  &\quad+O_{n,\varepsilon}(\gamma^{-1/2}).
  \label{eq:R-eps-ens-direct}
\end{align}
\end{theorem}

\begin{proof}
See Appendix~\ref{app:first-order-load-proof}.
\end{proof}

The expansion in \eqref{eq:R-eps-ens-direct} separates the leading
geometric growth from the constant-order reliability dependence.  The
prelog $(n-1)/(2n)$ is determined by the $n-1$ local dimensions of the
sphere, whereas
$n^{-1}\log_2\mu_{n,\varepsilon}^{\mathrm{ens}}$ contains all
dependence on the target error probability at constant order and
retains its blocklength dependence through the critical error curve.
The $O(\gamma^{-1/2})$ remainder follows from the quantitative
convergence of the finite-SNR error curve in a neighborhood of its
unique inverse.

More generally, suppose
\begin{equation}
  R(\gamma)
  =\frac{n-1}{2n}\log_2\gamma
  +\frac{\eta}{n}+o(1)
  \label{eq:rate-offset-scaling}
\end{equation}
and $M(\gamma)=2^{nR(\gamma)}$ is integer-valued.  Then
\begin{equation}
  \overline\Pe(M(\gamma),n;\gamma)
  \longrightarrow F_n^{\mathrm{ens}}(2^\eta).
  \label{eq:rate-offset-error-limit}
\end{equation}
Here $\eta$ is the constant-order offset in the block log-size
$\log_2M(\gamma)$, while $\eta/n$ is the corresponding rate offset per
real channel use.  Every finite $\eta$ selects the limiting error
$F_n^{\mathrm{ens}}(2^\eta)$, and the target error probability
$\varepsilon$ requires
$\eta=\log_2\mu_{n,\varepsilon}^{\mathrm{ens}}$.

\section{Equal-Energy Converse and High-SNR Rate Comparison}
\label{sec:deterministic-converse-comparison}

Ensemble averaging guarantees a deterministic spherical codebook whose
error probability does not exceed the ensemble average.  To assess its
high-SNR optimality, we compare this achievable rate with a converse
that applies to every deterministic equal-energy codebook under
average error probability.  The comparison remains within the equal-energy
class specified in Section~\ref{sec:model}.

For $0\le\theta\le\pi$, define the normalized solid angle
\begin{equation}
  \Omega_n(\theta)
  =
  \frac{\displaystyle\int_0^\theta\sin^{n-2}\phi\,\dd\phi}
       {\displaystyle\int_0^\pi\sin^{n-2}\phi\,\dd\phi}.
  \label{eq:normalized-cone-angle}
\end{equation}
Define
\begin{equation}
  \kappa_n
  =\frac{\Gamma(n/2)}
  {(n-1)\sqrt\pi\,\Gamma((n-1)/2)}.
  \label{eq:kappa-n-def}
\end{equation}
This is the leading small-angle coefficient:
$\Omega_n(\theta)\sim\kappa_n\theta^{n-1}$ as $\theta\downarrow0$.
Let $\mathbf e_1$ be the first standard basis vector in $\R^n$.  Let
$\theta_{n,M}$ be the unique angle satisfying
$\Omega_n(\theta_{n,M})=1/M$ and set
\begin{equation}
  B_n(M,\gamma)
  =
  \Prob\left\{
  \angle(\sqrt{n\gamma}\,\mathbf e_1+\mathbf Z,\mathbf e_1)
  >\theta_{n,M}
  \right\}.
  \label{eq:cone-escape-def}
\end{equation}
Thus $\theta_{n,M}$ is the half-angle of a circular cone occupying the
fraction $1/M$ of all directions, and $B_n(M,\gamma)$ is the probability
that noise drives the received direction outside this cone.
The classical circular-cone converse
\cite{GaussianOptimalCodes1959}, equivalently recovered by the optimized
minimax converse in this setting \cite{PolyanskiyMinimax2013}, states
that every deterministic equal-energy codebook satisfies
\begin{equation}
  \Pe(\mathcal C;\gamma)\ge B_n(M,\gamma),
  \qquad |\mathcal C|=M.
  \label{eq:pointwise-cone-bound-main}
\end{equation}
We evaluate this known inequality only in the fixed-$n$ high-SNR regime
needed for the rate comparison.

For $0<\varepsilon<1$, define
\begin{align}
  M_{\varepsilon,\mathrm{conv}}(n,\gamma)
  &=\max\{M\in\mathbb N:B_n(M,\gamma)\le\varepsilon\},
  \label{eq:M-cone-def}\\
  R_{\varepsilon,\mathrm{conv}}(n,\gamma)
  &=\frac1n\log_2M_{\varepsilon,\mathrm{conv}}(n,\gamma).
  \label{eq:R-conv-def}
\end{align}
For $r\ge1$ and $t\ge0$, let
\begin{equation}
  Q_r(t)=\Prob\{\chi_r^2>t\},
\end{equation}
where $\chi_r^2$ denotes a chi-square random variable with $r$ degrees
of freedom.  For $0<\varepsilon<1$, let $s_{r,\varepsilon}$ be the
unique positive solution of
\begin{equation}
  Q_r(s_{r,\varepsilon})=\varepsilon.
  \label{eq:chi-square-quantile-def}
\end{equation}
Finally, define
\begin{equation}
  \mu_{n,\varepsilon}^{\mathrm{conv}}
  =
  \frac{n^{(n-1)/2}}
  {\kappa_n s_{n-1,\varepsilon}^{(n-1)/2}}.
  \label{eq:mu-eps-conv-def}
\end{equation}
The factor $n^{(n-1)/2}$ reflects the signal radius
$\sqrt{n\gamma}$ associated with SNR $\gamma$ per real channel use.

The maximum defining $M_{\varepsilon,\mathrm{conv}}(n,\gamma)$ is well
defined: $B_n(1,\gamma)=0$ and
$B_n(M,\gamma)\uparrow1$ as $M\to\infty$, so the defining set is
nonempty and finite.

\begin{theorem}[High-SNR circular-cone bound and rate comparison]
\label{thm:equal-energy-rate-comparison}
Fix $n\ge2$ and $0<\varepsilon<1$.  If
$\mu(M(\gamma),\gamma)\to\mu\in(0,\infty)$, then
\begin{equation}
  B_n(M(\gamma),\gamma)
  \longrightarrow
  F_n^{\mathrm{conv}}(\mu)
  =
  Q_{n-1}\left(
  n(\kappa_n\mu)^{-2/(n-1)}
  \right).
  \label{eq:cone-error-curve}
\end{equation}
On the same scaled-codebook-size axis,
\begin{equation}
  F_n^{\mathrm{conv}}(\mu)
  <F_n^{\mathrm{ens}}(\mu),
  \qquad \mu>0.
  \label{eq:limiting-error-curve-ordering}
\end{equation}
Consequently,
\begin{equation}
  \mu_{n,\varepsilon}^{\mathrm{ens}}
  <
  \mu_{n,\varepsilon}^{\mathrm{conv}}.
  \label{eq:fixed-error-scale-ordering}
\end{equation}
Moreover, as $\gamma\to\infty$,
\begin{align}
  R_{\varepsilon,\mathrm{conv}}(n,\gamma)
  &=
  \frac{n-1}{2n}\log_2\gamma
  +\frac1n\log_2
  \mu_{n,\varepsilon}^{\mathrm{conv}}
  \nonumber\\
  &\quad+
  O_{n,\varepsilon}
  \left(
  \gamma^{-\min\{1,(n-1)/2\}}
  \right).
  \label{eq:R-conv-expansion}
\end{align}
For every $\gamma>0$,
\begin{equation}
  R_{\varepsilon,\mathrm{ens}}(n,\gamma)
  \le R_{\varepsilon,\mathrm{conv}}(n,\gamma).
  \label{eq:exact-achievability-converse-bounds}
\end{equation}
Consequently,
\begin{equation}
  \lim_{\gamma\to\infty}
  \frac{R_{\varepsilon,\mathrm{ens}}(n,\gamma)}
       {\log_2\gamma}
  =
  \lim_{\gamma\to\infty}
  \frac{R_{\varepsilon,\mathrm{conv}}(n,\gamma)}
       {\log_2\gamma}
  =\frac{n-1}{2n}.
  \label{eq:common-rate-prelog}
\end{equation}
For all sufficiently large $\gamma$,
\begin{equation}
  0\le
  1-
  \frac{R_{\varepsilon,\mathrm{ens}}(n,\gamma)}
       {R_{\varepsilon,\mathrm{conv}}(n,\gamma)}
  =O_{n,\varepsilon}
  \left(\frac1{\ln\gamma}\right).
  \label{eq:quantitative-achievability-converse-ratio}
\end{equation}
\end{theorem}

\begin{proof}
The high-SNR analysis and curve comparison are given in
Appendix~\ref{app:deterministic-converse-expansions}.  Ensemble
averaging supplies a deterministic codebook of size
$M_{\varepsilon,\mathrm{ens}}(n,\gamma)$ whose error probability does
not exceed $\varepsilon$.  The pointwise converse
\eqref{eq:pointwise-cone-bound-main} then gives
$B_n(M_{\varepsilon,\mathrm{ens}}(n,\gamma),\gamma)\le\varepsilon$,
so $M_{\varepsilon,\mathrm{ens}}(n,\gamma)\le
M_{\varepsilon,\mathrm{conv}}(n,\gamma)$ and
\eqref{eq:exact-achievability-converse-bounds} follows.  The two rate
expansions give \eqref{eq:common-rate-prelog}.  They also give
\[
  R_{\varepsilon,\mathrm{conv}}(n,\gamma)
  -R_{\varepsilon,\mathrm{ens}}(n,\gamma)=O_{n,\varepsilon}(1),
\]
whereas
$R_{\varepsilon,\mathrm{conv}}(n,\gamma)\sim
[(n-1)/(2n)]\log_2\gamma$.  This proves
\eqref{eq:quantitative-achievability-converse-ratio}.
\end{proof}

The theorem sandwiches the optimal fixed-error rate within the
deterministic equal-energy, average-error class between an achievable
rate and a converse bound having the same prelog $(n-1)/(2n)$.
Random spherical codes are therefore first-order optimal within this
class in the fixed-$n$, high-SNR regime.  Geometrically, at critical loading
the cone half-angle is of order $\gamma^{-1/2}$, and escape is governed
by the limiting transverse chi-square noise energy in
\eqref{eq:cone-error-curve}.  Because the rate prelogs coincide, the
strict ordering of the two limiting error curves affects the fixed-error
rates through their constant offsets: the converse offset is larger
than the ensemble-achievable offset, producing the nonzero limiting gap
analyzed next.

\section{Achievability--Converse Rate Gap and
Blocklength--Reliability Tradeoffs}
\label{sec:gap-design}

The preceding rate bounds share the same high-SNR prelog.  Their
difference converges to an SNR-independent constant that depends only
on $(n,\varepsilon)$.  We analyze this constant next.

\subsection{High-SNR achievability--converse gap}

Define the gap function as
\begin{equation}
  G(n,\varepsilon)
  =
  \lim_{\gamma\to\infty}
  \left[
  R_{\varepsilon,\mathrm{conv}}(n,\gamma)
  -
  R_{\varepsilon,\mathrm{ens}}(n,\gamma)
  \right].
  \label{eq:gap-def}
\end{equation}
The gap $G(n,\varepsilon)$ is measured in bits per real channel use.
Let $t_{a,\varepsilon}$ be the upper $\varepsilon$-quantile of
$T\sim\operatorname{Gamma}(a,1)$:
\begin{equation}
  \Prob\{T>t_{a,\varepsilon}\}=\varepsilon.
  \label{eq:gamma-quantile-def}
\end{equation}
Define
\begin{align}
  \mathcal L_a(x)
  &=
  \E\!\left[e^{-xT^a}\right]
  \nonumber\\
  &=
  \frac1{\Gamma(a)}
  \int_0^\infty
  u^{a-1}e^{-u-xu^a}\,\dd u,
  \qquad x\ge0.
  \label{eq:generalized-gamma-Laplace}
\end{align}
The function $\mathcal L_a$ is continuous and strictly decreasing from
one to zero.  We denote its functional inverse on $(0,1)$ by
$\mathcal L_a^{-1}$.
Since $2T\sim\chi_{n-1}^2$, the quantiles introduced in the converse
section satisfy
$s_{n-1,\varepsilon}=2t_{a,\varepsilon}$.  The ratio
\begin{equation}
  r_{a,\varepsilon}
  =
  \frac{\mu_{n,\varepsilon}^{\mathrm{ens}}}
       {\mu_{n,\varepsilon}^{\mathrm{conv}}}
  \label{eq:scaled-size-ratio}
\end{equation}
is the ratio of the ensemble-average fixed-error threshold constant to
the converse-bound threshold constant.  Using
$D_n=(2/n)^a\kappa_n$, the common geometric factor $D_n^{-1}$ cancels
from this ratio.

\begin{proposition}[High-SNR rate-bound gap]
\label{prop:high-snr-rate-gap}
For every $n\ge2$ and $0<\varepsilon<1$,
$r_{a,\varepsilon}\in(0,1)$ is the unique solution of
\begin{equation}
  \E\left[
  1-\exp\left\{
  -r_{a,\varepsilon}
  (T/t_{a,\varepsilon})^a
  \right\}
  \right]
  =\varepsilon.
  \label{eq:r-gap-root}
\end{equation}
Equivalently,
\begin{equation}
  r_{a,\varepsilon}
  =
  t_{a,\varepsilon}^{a}
  \mathcal L_a^{-1}(1-\varepsilon).
  \label{eq:r-gap-inverse-transform}
\end{equation}
Consequently, the gap has the exact representation
\begin{equation}
  G(n,\varepsilon)
  =
  -\frac1n\log_2\!\left[
  t_{a,\varepsilon}^{a}
  \mathcal L_a^{-1}(1-\varepsilon)
  \right].
  \label{eq:gap-inverse-transform}
\end{equation}
Equivalently,
\begin{equation}
  G(n,\varepsilon)
  =
  -\frac1n\log_2r_{a,\varepsilon}
  =
  \frac1n\log_2
  \frac{\mu_{n,\varepsilon}^{\mathrm{conv}}}
       {\mu_{n,\varepsilon}^{\mathrm{ens}}}
  >0.
  \label{eq:gap-identity}
\end{equation}
As $\gamma\to\infty$,
\begin{equation}
  R_{\varepsilon,\mathrm{conv}}(n,\gamma)
  -
  R_{\varepsilon,\mathrm{ens}}(n,\gamma)
  =
  G(n,\varepsilon)
  +O_{n,\varepsilon}(\gamma^{-1/2}).
  \label{eq:finite-snr-gap-expansion}
\end{equation}
Moreover, the relative rate shortfall satisfies
\begin{equation}
  1-
  \frac{R_{\varepsilon,\mathrm{ens}}(n,\gamma)}
       {R_{\varepsilon,\mathrm{conv}}(n,\gamma)}
  \sim
  \frac{2nG(n,\varepsilon)}
       {(n-1)\log_2\gamma}.
  \label{eq:relative-shortfall-exact}
\end{equation}
\end{proposition}

\begin{proof}
See Appendix~\ref{app:rate-gap-asymptotics}.
\end{proof}

Because $r_{a,\varepsilon}$ is the ratio of the ensemble and converse
scaled-codebook thresholds, its negative base-2 logarithm divided by
$n$ is their limiting rate-bound gap.  Thus $G(n,\varepsilon)$ is the
exact limiting separation between the explicit ensemble-achievable rate and the
circular-cone converse rate bound.  Since the converse bound need not
be attainable, this separation should not be interpreted as a proven
coding loss.  Equation~\eqref{eq:relative-shortfall-exact} further
identifies the leading coefficient of the relative shortfall.

\subsection{Joint reliability--blocklength behavior}

Let $\Phi$ denote the standard normal distribution function, define
\[
  z_\varepsilon=\Phi^{-1}(1-\varepsilon),
\]
and let $\gamma_{\mathrm E}$ be the Euler--Mascheroni constant.

\begin{theorem}[Joint asymptotics of the rate-bound gap]
\label{thm:gap-reliability-blocklength}
The following limits describe three complementary asymptotic regimes.

\emph{1) Increasing reliability at fixed blocklength:}
For fixed $n\ge2$, let $L=\ln(1/\varepsilon)$.  As
$\varepsilon\downarrow0$,
\begin{equation}
  G(n,\varepsilon)
  =
  \frac{
  L-a\ln L+\ln\!\frac{\Gamma(2a)}{\Gamma(a)}
  }{n\ln2}
  +
  O_n\left(\frac{\ln L}{nL}\right).
  \label{eq:gap-small-epsilon}
\end{equation}
In particular,
\begin{equation}
  G(n,\varepsilon/10)-G(n,\varepsilon)
  \longrightarrow\frac{\log_2 10}{n},
  \qquad \varepsilon\downarrow0.
  \label{eq:fixed-n-decade-cost}
\end{equation}

\emph{2) Increasing blocklength at fixed target error probability:}
For fixed $0<\varepsilon<1$, as $n\to\infty$,
\begin{equation}
  G(n,\varepsilon)
  =
  \frac{\gamma_{\mathrm E}}{n\ln2}
  +
  \frac{\pi^2\sqrt2\,z_\varepsilon}
       {12\ln2\,n^{3/2}}
  +O_\varepsilon(n^{-2}).
  \label{eq:gap-large-n}
\end{equation}
Hence
\begin{equation}
  nG(n,\varepsilon)
  \longrightarrow\frac{\gamma_{\mathrm E}}{\ln2}.
  \label{eq:universal-nG-limit}
\end{equation}

\emph{3) Exponentially stringent reliability:}
As $n\to\infty$, suppose $\varepsilon_n\downarrow0$ and
\begin{equation}
  \frac1n\ln\frac1{\varepsilon_n}\longrightarrow\rho>0.
  \label{eq:joint-reliability-exponent}
\end{equation}
Let $q_\rho>1$ be the unique solution of
\begin{equation}
  q_\rho-1-\ln q_\rho=2\rho.
  \label{eq:q-rho-def}
\end{equation}
If $q_\rho\le2$, then $G(n,\varepsilon_n)\to0$.  If $q_\rho>2$, then
\begin{equation}
  G(n,\varepsilon_n)
  \longrightarrow
  \frac{
  q_\rho-2-2\ln(q_\rho/2)
  }{2\ln2}
  >0.
  \label{eq:G-supercritical-reliability}
\end{equation}
The boundary occurs at
\begin{equation}
  \rho_{\mathrm c}=\frac{1-\ln2}{2}.
  \label{eq:critical-reliability-exponent}
\end{equation}
\end{theorem}

\begin{proof}
See Appendix~\ref{app:rate-gap-asymptotics}.
\end{proof}

At fixed $n$, replacing $\varepsilon$ by $\varepsilon/10$ increases the
rate-bound gap asymptotically by $\log_2(10)/n$ bits per real channel
use.  At fixed $\varepsilon$, the gap decays asymptotically as $1/n$,
with reliability first appearing in the $n^{-3/2}$ correction.  Under
the joint scaling $\varepsilon_n=e^{-\rho n+o(n)}$, the gap vanishes for
$0<\rho\le\rho_{\mathrm c}$ and converges to the positive value in
\eqref{eq:G-supercritical-reliability} for $\rho>\rho_{\mathrm c}$.
The value $q_\rho=2$ is the integrability threshold in the transformed
gap equation: its kernel is linear near the origin, whereas the
limiting weight is proportional to $x^{-q_\rho}$.  At and below this
threshold, the ensemble and converse scaled-codebook thresholds are
separated only subexponentially in $n$, so the rate-bound gap vanishes;
above it, their exponential separation leaves a positive limiting gap.
A positive limiting rate-bound gap does not establish that a higher
rate is achievable; it means only that the circular-cone converse
remains additively separated from the ensemble-achievable fixed-error
rate.

\begin{corollary}[Sequential closure of the rate-bound gap]
\label{cor:sequential-gap-closure}
For every fixed $0<\varepsilon<1$,
\begin{equation}
  \lim_{n\to\infty}\lim_{\gamma\to\infty}
  \left[
  R_{\varepsilon,\mathrm{conv}}(n,\gamma)
  -
  R_{\varepsilon,\mathrm{ens}}(n,\gamma)
  \right]
  =0.
  \label{eq:sequential-gap-closure}
\end{equation}
Hence the explicit achievability and converse bounds become additively
tight when the high-SNR limit is taken before the large-blocklength
limit.
\end{corollary}

\begin{proof}
The inner limit equals $G(n,\varepsilon)$ by \eqref{eq:gap-def}, and
\eqref{eq:gap-large-n} gives $G(n,\varepsilon)\to0$.
\end{proof}

This sequential statement does not imply a uniform approximation when
$n$ and $\gamma$ grow simultaneously.

\subsection{Blocklength required for a prescribed gap}

We invert the gap formulas in two complementary design limits: first,
$G_0\downarrow0$ at fixed reliability, and second,
$\varepsilon\downarrow0$ at fixed permitted gap $G_0$.

For a target gap $G_0>0$, measured in bits per real channel use, define
\begin{equation}
  n_G(\varepsilon,G_0)
  =
  \min\{n\in\mathbb N:n\ge2,\ G(n,\varepsilon)\le G_0\}.
  \label{eq:nG-def}
\end{equation}
Equation~\eqref{eq:universal-nG-limit} guarantees that the defining set
is nonempty.  The quantity $n_G$ controls the limiting rate-bound gap
obtained after letting $\gamma\to\infty$.  Because the remainder in
\eqref{eq:finite-snr-gap-expansion} is not uniform in $n$, this
definition does not by itself provide a finite-SNR guarantee when
$n_G$ grows.

\begin{proposition}[Blocklength required for a prescribed gap]
\label{prop:prescribed-gap-blocklength}
\emph{1) Fixed reliability:}
For fixed $0<\varepsilon<1$, put
\[
  N_0(G_0)=\frac{\gamma_{\mathrm E}}{G_0\ln2},
  \qquad
  c_0=\frac{\pi^2\sqrt2}{12}.
\]
As $G_0\downarrow0$,
\begin{equation}
  n_G(\varepsilon,G_0)
  =
  N_0(G_0)
  +
  \frac{c_0z_\varepsilon}{\gamma_{\mathrm E}}
  \sqrt{N_0(G_0)}
  +O_\varepsilon(1).
  \label{eq:nG-small-gap-law}
\end{equation}
Therefore,
\begin{align}
  &n_G(\varepsilon/10,G_0)
  -
  n_G(\varepsilon,G_0)
  \nonumber\\
  &\quad=
  \frac{c_0}{\gamma_{\mathrm E}}
  \bigl(z_{\varepsilon/10}-z_\varepsilon\bigr)
  \sqrt{N_0(G_0)}
  +O_\varepsilon(1).
  \label{eq:nG-decade-law}
\end{align}
\emph{2) Increasing reliability:}
For fixed gap tolerance $G_0>0$, let
$q_{G_0}>2$ solve
\[
  q_{G_0}-2-2\ln(q_{G_0}/2)=2G_0\ln2
\]
and define
\[
  \rho_{G_0}
  =
  \frac{q_{G_0}-1-\ln q_{G_0}}2.
\]
Then
\begin{equation}
  n_G(\varepsilon,G_0)
  \sim
  \frac{\ln(1/\varepsilon)}{\rho_{G_0}},
  \qquad \varepsilon\downarrow0.
  \label{eq:nG-ultrareliable-law}
\end{equation}
\end{proposition}

\begin{proof}
See Appendix~\ref{app:blocklength-design-proofs}.
\end{proof}

The leading requirement
$N_0(G_0)=\gamma_{\mathrm E}/(G_0\ln2)$ is independent of the fixed
target error probability and grows inversely with the permitted gap.
Reliability first enters through the order-$\sqrt{N_0(G_0)}$
correction; since
$z_{\varepsilon/10}>z_\varepsilon$, replacing $\varepsilon$ by $\varepsilon/10$
requires an asymptotically positive blocklength increment at this order
as $G_0\downarrow0$.  In the complementary ultra-reliable limit, the
leading-order law $\ln(1/\varepsilon)/\rho_{G_0}$ predicts an increment
of $\ln10/\rho_{G_0}$ channel uses per reliability decade.  This is a
first-order design interpretation; the integer-valued increment itself
requires a sharper remainder estimate.

The prescribed-gap problem controls the limiting separation between
the two rate bounds and completes the analytical
blocklength--reliability tradeoff developed in this section.

\section{Numerical Results}
\label{sec:numerical-results}

This section illustrates the error-probability transition, the
fixed-error data rates, and the achievability--converse rate-bound gap.

\subsection{Numerical evaluation and Monte Carlo validation}

The finite-SNR ensemble-average error probability is evaluated from the
exact dangerous-cap representation in
\eqref{eq:dangerous-identity}, and the ensemble-achievable rate and
converse rate bound are obtained by monotone inversion of their
respective error probabilities.  Smooth curves use continuous
codebook-size extensions
only where integer rounding is graphically negligible, whereas every
simulation marker corresponds to an integer codebook size.  Conditional
Monte Carlo averages the exact error probability conditioned on the
noise and checks the numerical integration and inversion.  Direct ML
simulation explicitly generates the spherical competitor codewords and
provides an independent end-to-end check of the dangerous-cap
reduction.  All
displayed SNR values are $10\log_{10}\gamma$ in decibels, where
$\gamma$ is the SNR per real channel use.  Circles denote conditional
Monte Carlo estimates; the open squares in
Fig.~\ref{fig:numerical-error-transition} denote direct ML Monte Carlo
estimates for the spherical ensemble.

\subsection{Error-transition scaling}

Fig.~\ref{fig:numerical-error-transition} shows the
ensemble-average error probability as a function of the scaled
codebook size
\[
  \mu=\frac{M-1}{\gamma^{(n-1)/2}}
\]
for $n=6$.  At 5 dB, a visible finite-SNR displacement from
$F_n^{\mathrm{ens}}(\mu)$ remains.  The 20- and 40-dB curves are
progressively closer to the limit, and the 40-dB curve is visually
indistinguishable from it over the displayed transition.  The
agreement occurs on the scaled-size axis rather than on the
unnormalized codebook-size axis, confirming that
$M\asymp\gamma^{(n-1)/2}$ is the critical codebook-growth order.

The figure also displays both endpoints of the transition.  For small
$\mu$, the error probability approaches zero, consistently with the
underloaded regime in \eqref{eq:underload-regime}; for large $\mu$, it
approaches one, consistently with the overloaded regime in
\eqref{eq:overload-regime}.  Between these endpoints, the nontrivial
transition agrees with \eqref{eq:critical-regime}.  The conditional
Monte Carlo and direct ML markers follow the corresponding finite-SNR
curves throughout the accessible transition region.

\begin{figure}[t]
  \centering
  \includegraphics[width=\columnwidth]
  {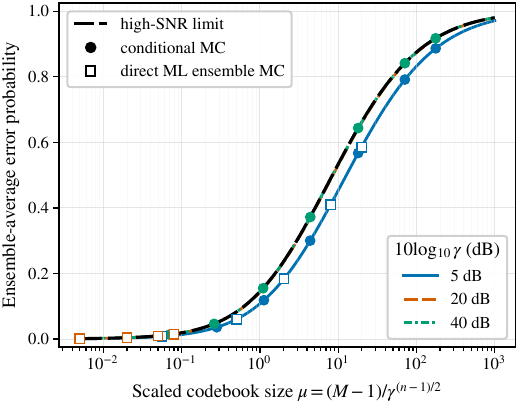}
  \caption{Ensemble-average error probability versus the scaled
  codebook size $\mu=(M-1)/\gamma^{(n-1)/2}$ for $n=6$.  Colored
  curves show the finite-SNR numerical evaluations at 5, 20, and
  40 dB SNR per real channel use, and the black dashed curve is
  $F_n^{\mathrm{ens}}(\mu)$.  Circles are conditional Monte Carlo
  estimates, while open squares are direct ML Monte Carlo estimates
  using integer codebook sizes.  The markers are plotted at
  their realized values of $\mu$.}
  \label{fig:numerical-error-transition}
\end{figure}

\subsection{Fixed-error rates}

Fig.~\ref{fig:numerical-fixed-error-rates} compares
$R_{\varepsilon,\mathrm{ens}}(n,\gamma)$ with the converse rate bound
$R_{\varepsilon,\mathrm{conv}}(n,\gamma)$.  In panel~(a), the
blocklength is fixed at $n=16$ while the target error probability is
varied.  Increasing the reliability requirement lowers both the
ensemble-achievable rate and the converse rate bound
through the constant high-SNR offsets
$n^{-1}\log_2\mu_{n,\varepsilon}^{\mathrm{ens}}$ and
$n^{-1}\log_2\mu_{n,\varepsilon}^{\mathrm{conv}}$, while their leading
slopes remain unchanged.  This behavior is consistent with
\eqref{eq:R-eps-ens-direct} and \eqref{eq:R-conv-expansion}: the target
error probability enters the constant-order term but not the common
prelog $(n-1)/(2n)$.

Panel~(b) fixes $\varepsilon=10^{-2}$ and varies the blocklength.  The
ensemble-achievable rates and converse rate bounds are ordered, from
highest to lowest, by $n=32$, $n=8$, and $n=4$ throughout the displayed
SNR range, while the high-SNR slope approaches $1/2$
as $n$ increases.  The solid--dashed separation also narrows with $n$,
anticipating the large-blocklength gap behavior in
\eqref{eq:gap-large-n}.  In both panels the ensemble-achievable rate
remains below the converse rate bound, as required by
\eqref{eq:exact-achievability-converse-bounds}, and the conditional
Monte Carlo rate estimates agree with the ensemble curves.

The small staircase at the lowest displayed SNRs, most visible for
$n=4$, arises from the operational integer-codebook-size constraint.

\begin{figure*}[t]
  \centering
  \includegraphics[width=\textwidth]
  {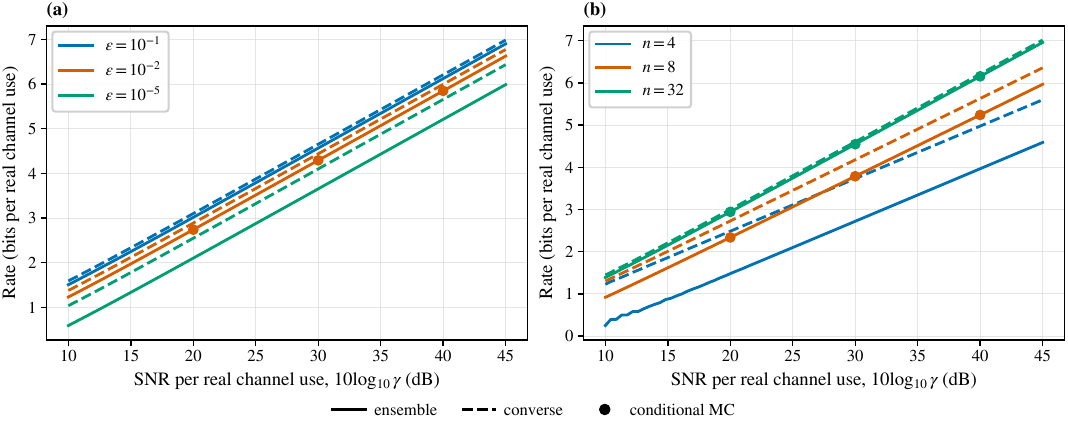}
  \caption{Finite-SNR ensemble-achievable rate and converse rate bound
  versus SNR per real channel use.  Solid curves show
  $R_{\varepsilon,\mathrm{ens}}$, dashed curves are
  the converse rate bound $R_{\varepsilon,\mathrm{conv}}$, and circles
  are conditional Monte
  Carlo ensemble-rate estimates.  (a) Fixed $n=16$ and target error
  probabilities $\varepsilon\in\{10^{-1},10^{-2},10^{-5}\}$.
  (b) Fixed $\varepsilon=10^{-2}$ and blocklengths
  $n\in\{4,8,32\}$.}
  \label{fig:numerical-fixed-error-rates}
\end{figure*}

\subsection{Finite-SNR and limiting rate-bound gaps}

To distinguish the finite-SNR rate difference from its limiting value,
define
\begin{equation}
  \Delta_R(n,\varepsilon,\gamma)
  =
  R_{\varepsilon,\mathrm{conv}}(n,\gamma)
  -
  R_{\varepsilon,\mathrm{ens}}(n,\gamma).
  \label{eq:numerical-finite-snr-gap}
\end{equation}
Proposition~\ref{prop:high-snr-rate-gap} gives
\[
  \Delta_R(n,\varepsilon,\gamma)
  =
  G(n,\varepsilon)+O_{n,\varepsilon}(\gamma^{-1/2}).
\]
Fig.~\ref{fig:numerical-rate-difference} compares
$\Delta_R(n,\varepsilon,\gamma)$ with the corresponding horizontal
level $G(n,\varepsilon)$.

Panel~(a) fixes $n=16$.  More stringent reliability produces a larger
limiting rate-bound gap, in agreement with the fixed-$n$ reliability
dependence in Theorem~\ref{thm:gap-reliability-blocklength}.  Panel~(b)
fixes $\varepsilon=10^{-2}$.  The gap decreases markedly as the
blocklength increases, consistently with the large-$n$ expansion in
\eqref{eq:gap-large-n}.  In both panels, the finite-SNR differences
approach their respective dotted limiting levels, and the conditional
Monte Carlo points agree with the numerical rate differences.

The visible jumps and oscillations at the lowest displayed SNR values,
most prominently for $\varepsilon=10^{-4}$ in panel~(a) and $n=4$ in
panel~(b), result
from the integer thresholds defining the two finite-SNR rates.  The
difference of two quantized rate thresholds can change nonmonotonically
when either admissible codebook size increases by one.  Their
magnitude rapidly becomes negligible as the codebook sizes grow.
\begin{figure*}[t]
  \centering
  \includegraphics[width=\textwidth]
  {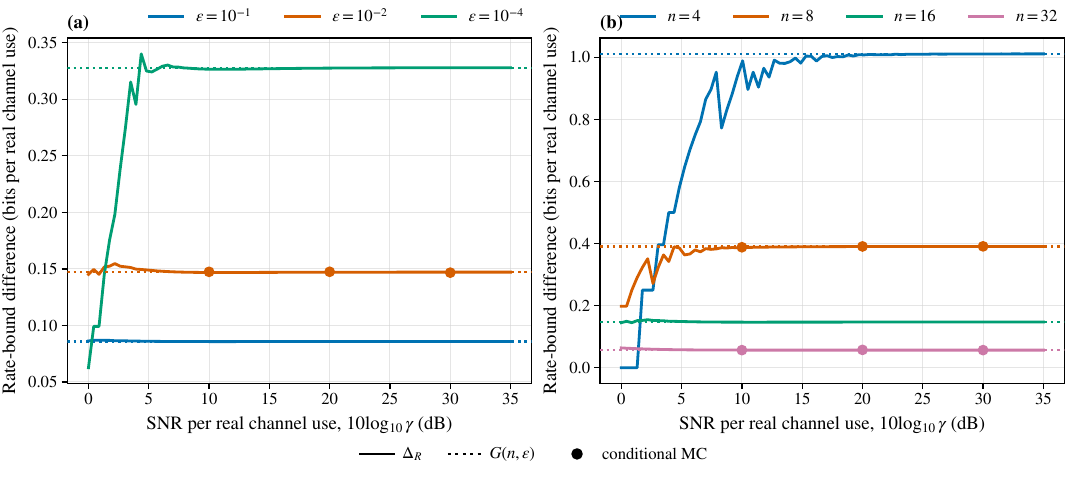}
  \caption{Finite-SNR rate-bound difference
  $\Delta_R(n,\varepsilon,\gamma)$ versus SNR per real channel use,
  together with its limiting value $G(n,\varepsilon)$.  Solid curves
  show $\Delta_R$, color-matched
  dotted lines show $G(n,\varepsilon)$, and circles are conditional
  Monte Carlo estimates of the finite-SNR difference.  (a) Fixed
  $n=16$ and
  $\varepsilon\in\{10^{-1},10^{-2},10^{-4}\}$.  (b) Fixed
  $\varepsilon=10^{-2}$ and $n\in\{4,8,16,32\}$.  The low-SNR
  oscillations are caused by integer-codebook-size quantization.}
  \label{fig:numerical-rate-difference}
\end{figure*}

Fig.~\ref{fig:numerical-limiting-gap} examines the dependence of the
limiting gap on blocklength and reliability without a finite-SNR
remainder.  Panel~(a) shows $G(n,\varepsilon)$ over
$2\le n\le 10^3$.  The log-log axes expose the leading $1/n$ decay,
and the curves approach the asymptotic law
$\gamma_{\mathrm E}/(n\ln2)$.  At a fixed finite blocklength, a smaller
target error probability produces a larger gap.

Panel~(b) shows $nG(n,\varepsilon)$ over
$50\le n\le10^4$.  For every displayed fixed
$\varepsilon$, the curves approach the universal limit
$\gamma_{\mathrm E}/\ln2$ in \eqref{eq:universal-nG-limit}.  The
different approach rates are explained by the
$n^{-3/2}$ reliability-dependent term in \eqref{eq:gap-large-n}.
Consequently, stringent reliability can delay the onset of the
universal leading behavior even though it does not change the limit.
These curves concern fixed $\varepsilon$ as $n$ increases; they do not
represent the joint scaling $\varepsilon=\varepsilon_n$ in
\eqref{eq:joint-reliability-exponent}.

\begin{figure*}[t]
  \centering
  \includegraphics[width=\textwidth]
  {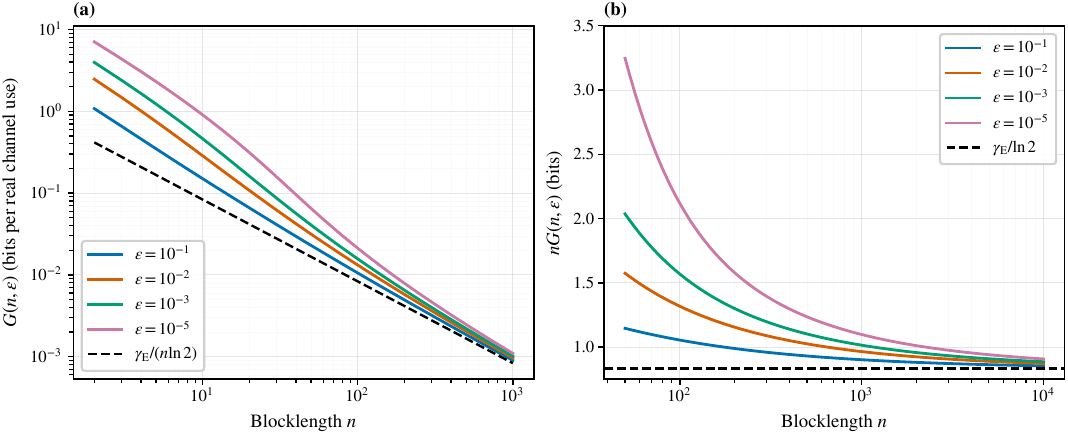}
  \caption{Limiting rate-bound gap as a function of blocklength for
  fixed target error probabilities
  $\varepsilon\in\{10^{-1},10^{-2},10^{-3},10^{-5}\}$.
  (a) $G(n,\varepsilon)$ together with the leading approximation
  $\gamma_{\mathrm E}/(n\ln2)$.  (b) The scaled gap
  $nG(n,\varepsilon)$ together with its universal limit
  $\gamma_{\mathrm E}/\ln2$, where $\gamma_{\mathrm E}$ is the
  Euler--Mascheroni constant.}
  \label{fig:numerical-limiting-gap}
\end{figure*}

\section{Conclusion}
\label{sec:conclusion}

This paper characterizes random spherical codebooks in the
fixed-blocklength, high-SNR regime in which the SNR per real channel use
tends to infinity and the codebook size grows with it.  The governing
mechanism is the balance between a shrinking noise-selected dangerous
cap and a growing number of competitor codewords.  This balance
identifies the critical codebook-growth order and separates the
underloaded, critically loaded, and overloaded error regimes.  On the
critical scale, the fixed-error data rate uses the sphere's $n-1$
tangential dimensions.  Its high-SNR prelog agrees with that of the
circular-cone converse for the deterministic equal-energy,
average-error class, establishing first-order optimality of random
spherical codes within this class.

First-order agreement does not imply additive agreement.  At fixed
blocklength and target error probability, the ensemble-achievable rate
remains separated from the converse rate bound by a positive limiting
gap.  At fixed target error probability, this gap vanishes as $1/n$ in
the sequential high-SNR and large-blocklength limit, with a universal
leading coefficient.  Under an exponentially decreasing target error
probability, the joint analysis identifies the reliability threshold
between vanishing and nonvanishing limiting gaps and quantifies the
blocklength required to meet a prescribed gap.  The fixed-$n$ rate and
gap expansions are not asserted to hold uniformly when $n$ and
$\gamma$ grow simultaneously.

\appendices

\section{Scaled Dangerous-Cap Limit}
\label{app:local-mass-limits}

This appendix proves Proposition~\ref{prop:local-mass-limit}.

Set $\delta=\gamma^{-1/2}$, rotate coordinates so that the transmitted
codeword is $\mathbf e_1$, and couple all SNR values through the same
noise vector.  Write
\[
  \mathbf Z=Z_\parallel\mathbf e_1+\mathbf Z_\perp,
  \qquad
  T_\perp=\frac{\|\mathbf Z_\perp\|^2}{2},
  \qquad
  L_\delta=1+\frac{\delta}{\sqrt n}Z_\parallel.
\]
Then $T_\perp\sim\operatorname{Gamma}(a,1)$ and is independent of
$Z_\parallel$.  By \eqref{eq:Hgamma-def}, the exact dangerous-cap
height is
\[
  h_\delta
  =
  1-\frac{L_\delta}
  {\sqrt{L_\delta^2+(2\delta^2/n)T_\perp}}.
\]
The beta representation gives
\begin{equation}
  \begin{aligned}
    I_{h/2}(a,a)
    &=
    \frac{h^a}{2^a a\,\mathrm B(a,a)}
    \bigl(1+O_n(h)\bigr)
    \\
    &=n^aD_nh^a\bigl(1+O_n(h)\bigr),
    \qquad h\downarrow0.
  \end{aligned}
  \label{eq:small-cap-prob-app}
\end{equation}
For the convergence rate, use the good event
\[
  \mathcal H_\delta
  =
  \{|Z_\parallel|\le\delta^{-1/2},\ 
  T_\perp\le\delta^{-1}\}.
\]
For all sufficiently small $\delta$, write on $\mathcal H_\delta$
\[
  S_\delta=\sqrt{L_\delta^2+(2\delta^2/n)T_\perp}.
\]
Then
\[
  h_\delta=\frac{\delta^2}{n}T_\perp g_\delta,
  \qquad
  g_\delta=\frac{2}{S_\delta(S_\delta+L_\delta)}.
\]
Here $L_\delta=1+O_n(\delta^{1/2})$ and
$\delta^2T_\perp/n=O_n(\delta)$.  Smoothness near $(L_\delta,
\delta^2T_\perp/n)=(1,0)$, together with
\eqref{eq:small-cap-prob-app}, therefore implies
\begin{align*}
  &\left|
  \delta^{-2a}\Pi_\gamma-D_nT_\perp^a
  \right|\mathbf 1_{\mathcal H_\delta}
  \\
  &\quad\le
  C_nT_\perp^a
  \bigl(\delta|Z_\parallel|+\delta^2T_\perp\bigr).
\end{align*}
For $q\in\{1,2\}$, the $L^q$ norm of the right-hand side is
$O_{n,q}(\delta)$.  On $\mathcal H_\delta^c$, use
$0\le\Pi_\gamma\le1$ together with Gaussian and gamma Chernoff bounds
to obtain, for every fixed $b>0$,
\[
  \left\|\delta^{-2a}\Pi_\gamma
  \mathbf 1_{\mathcal H_\delta^c}\right\|_q
  +
  \left\|D_nT_\perp^a
  \mathbf 1_{\mathcal H_\delta^c}\right\|_q
  =o(\delta^b).
\]
Combining the good- and bad-event estimates gives
\[
  \left\|
  \gamma^a\Pi_\gamma-D_nT_\perp^a
  \right\|_q
  =
  O_{n,q}(\gamma^{-1/2}).
\]
This proves \eqref{eq:local-mass-Lq}.  For $k\in\{1,2\}$,
\[
  \E[\Theta_n^k]
  =
  D_n^k\frac{\Gamma((k+1)a)}{\Gamma(a)}.
\]
Taking expectations for $k=1$ proves \eqref{eq:p1-leading}; the
$L^2$ convergence also gives
\[
  \gamma^{2a}\E[\Pi_\gamma^2]
  \longrightarrow\E[\Theta_n^2],
\]
which is used in the Bonferroni estimate in
Appendix~\ref{app:first-order-load-proof}.

\section{Proofs of the Error Transition and Fixed-Error Data Rate}
\label{app:first-order-load-proof}

This appendix proves Theorem~\ref{thm:load-scaled},
Corollary~\ref{cor:growing-codebook-regimes}, and
Theorem~\ref{thm:ensemble-fixed-error-rate}.  Put $m=M-1$ and
$\mu=\mu(M,\gamma)=m/\gamma^a$.

\subsection{Critical-scale approximation}

For $p\in[0,1]$ and every integer $m\ge1$, the zero-count case of
Le~Cam's inequality \cite{LeCam1960} gives
\[
  0\le e^{-mp}-(1-p)^m\le2mp^2.
\]
The case $m=0$ is immediate.  Moreover,
$|e^{-x}-e^{-y}|\le|x-y|$ for $x,y\ge0$.  Since
$m=\mu\gamma^a$, the exact error representation gives
\begin{align*}
  &\left|
  \overline\Pe(M,n;\gamma)-F_n^{\mathrm{ens}}(\mu)
  \right|
  \\
  &\quad\le
  2\mu\gamma^{-a}
  \E\!\left[(\gamma^a\Pi_\gamma)^2\right]
  \\
  &\qquad+
  \mu\left\|\gamma^a\Pi_\gamma-\Theta_n\right\|_1.
\end{align*}
By Proposition~\ref{prop:local-mass-limit}, the first term is
$O_{n,K}(\gamma^{-a})$ and the second is
$O_{n,K}(\gamma^{-1/2})$, uniformly for $0\le\mu\le K$.  Since
$a\ge1/2$, this proves
\eqref{eq:quantitative-error-curve}, and
\eqref{eq:load-error-limit} follows immediately.

\subsection{Underloaded and overloaded regimes}

Suppose first that $M(\gamma)\to\infty$ and $\mu_\gamma\to0$.
The first two Bonferroni inequalities applied to the competitor-codeword events
give
\[
  m\E[\Pi_\gamma]
  -\binom m2\E[\Pi_\gamma^2]
  \le\overline\Pe(M,n;\gamma)
  \le m\E[\Pi_\gamma].
\]
By \eqref{eq:p1-leading},
\[
  m\E[\Pi_\gamma]
  =A_n\mu_\gamma
  \bigl(1+O_n(\gamma^{-1/2})\bigr),
\]
whereas the $q=2$ part of
Proposition~\ref{prop:local-mass-limit} gives
\[
  \binom m2\E[\Pi_\gamma^2]
  =O_n(\mu_\gamma^2).
\]
Consequently,
\[
  \overline\Pe(M(\gamma),n;\gamma)
  =A_n\mu_\gamma
  \left[1+O_n(\gamma^{-1/2})+O_n(\mu_\gamma)\right],
\]
which proves \eqref{eq:underload-regime}.  The critical regime is the
finite positive case of Theorem~\ref{thm:load-scaled}.

Now suppose $\mu_\gamma\to\infty$.  Put
$V_\gamma=\gamma^a\Pi_\gamma$.  For any $\zeta>0$,
\begin{align*}
  \Prob\{N_\gamma=0\}
  &=\E[(1-\Pi_\gamma)^m]
  \\
  &\le
  \Prob\{V_\gamma\le\zeta\}
  +e^{-\zeta\mu_\gamma}.
\end{align*}
Since $V_\gamma\to\Theta_n$ in probability and $\Theta_n>0$ almost
surely,
\[
  \limsup_{\gamma\to\infty}\Prob\{N_\gamma=0\}
  \le\Prob\{\Theta_n\le\zeta\}.
\]
Letting $\zeta\downarrow0$ proves \eqref{eq:overload-regime} and
completes the proof of Corollary~\ref{cor:growing-codebook-regimes}.

Finally, differentiation under the expectation gives
\[
  \frac{\dd}{\dd\mu}F_n^{\mathrm{ens}}(\mu)
  =
  \E[\Theta_ne^{-\mu\Theta_n}]>0.
\]
Hence the curve is continuous and strictly increasing.  Its endpoint
values follow from dominated convergence and the fact that
$\Theta_n>0$ almost surely.

\subsection{Fixed-error inversion}

Let
$\mu_\varepsilon=\mu_{n,\varepsilon}^{\mathrm{ens}}$.
Continuity of
$F_n^{\mathrm{ens}\,\prime}(\mu)
=\E[\Theta_ne^{-\mu\Theta_n}]$ provides a neighborhood
$\mathcal U$ of $\mu_\varepsilon$ and a constant $c>0$ such that
$F_n^{\mathrm{ens}\,\prime}\ge c$ on $\mathcal U$.  Strict inequalities
at the endpoints of $\mathcal U$, together with
\eqref{eq:quantitative-error-curve}, place the threshold grid point
\[
  \widehat\mu_\gamma
  =\frac{M_{\varepsilon,\mathrm{ens}}(n,\gamma)-1}{\gamma^a}
\]
and its successor in $\mathcal U$ for all sufficiently large $\gamma$.  The
threshold definition and the uniform error bound imply
\begin{align*}
  F_n^{\mathrm{ens}}(\widehat\mu_\gamma)
  &\le\varepsilon+C_{n,\varepsilon}\gamma^{-1/2},\\
  F_n^{\mathrm{ens}}(\widehat\mu_\gamma+\gamma^{-a})
  &>\varepsilon-C_{n,\varepsilon}\gamma^{-1/2}.
\end{align*}
The mean-value theorem and
$\gamma^{-a}\le\gamma^{-1/2}$ therefore give
\[
  \widehat\mu_\gamma=\mu_\varepsilon
  +O_{n,\varepsilon}(\gamma^{-1/2}),
\]
which proves \eqref{eq:M-eps-leading} and
\eqref{eq:M-eps-leading-unscaled}.  To pass from the integer codebook
size to its rate, use the exact identity
\begin{align*}
  R_{\varepsilon,\mathrm{ens}}(n,\gamma)
  &=\frac an\log_2\gamma
  +\frac1n\log_2
  \left(\widehat\mu_\gamma+\gamma^{-a}\right).
\end{align*}
Since $\mu_\varepsilon>0$,
$\widehat\mu_\gamma=\mu_\varepsilon+O_{n,\varepsilon}
(\gamma^{-1/2})$, and $\gamma^{-a}\le\gamma^{-1/2}$, smoothness of
the logarithm near $\mu_\varepsilon$ proves
\eqref{eq:R-eps-ens-direct}.  Finally, if
\eqref{eq:rate-offset-scaling} holds, then
$(M(\gamma)-1)/\gamma^a\to2^\eta$, and
\eqref{eq:rate-offset-error-limit} follows from
Theorem~\ref{thm:load-scaled}.

\section{High-SNR Expansion of the Circular-Cone Rate Bound}
\label{app:deterministic-converse-expansions}

This appendix proves the high-SNR claims in
Theorem~\ref{thm:equal-energy-rate-comparison}.  The underlying
nonasymptotic converse is cited in the main text and is not rederived.

From \eqref{eq:normalized-cone-angle},
$\sin\phi=\phi(1+O(\phi^2))$ at the origin, and
$\int_0^\pi\sin^{n-2}\phi\,\dd\phi
=\sqrt\pi\,\Gamma((n-1)/2)/\Gamma(n/2)$, we obtain
\[
  \Omega_n(\theta)
  =\kappa_n\theta^{n-1}
  \left[1+O_n(\theta^2)\right].
\]
This supplies the error order needed below.

Uniformly whenever $M/\gamma^a$ ranges over a fixed compact set
$K\subset(0,\infty)$, inversion of the same expansion and
$\tan\theta=\theta(1+O(\theta^2))$ give
\begin{equation}
  \sqrt{n\gamma}\tan\theta_{n,M}
  =\sqrt n\left(\kappa_n\frac{M}{\gamma^a}\right)^{-1/(n-1)}
  \left[1+O_{n,K}(\gamma^{-1})\right].
  \label{eq:uniform-cone-angle-expansion-app}
\end{equation}

Suppose
$\mu(M_j,\gamma_j)\to\mu\in(0,\infty)$.  Then
$M_j/\gamma_j^a=\mu(M_j,\gamma_j)+\gamma_j^{-a}\to\mu$, and
\eqref{eq:uniform-cone-angle-expansion-app} gives
$\sqrt{n\gamma_j}\tan\theta_{n,M_j}\to
\sqrt n(\kappa_n\mu)^{-1/(n-1)}$.
Write
\[
  \mathbf Z=Z_\parallel\mathbf e_1+\mathbf Z_\perp,
  \qquad \mathbf Z_\perp\in\R^{n-1}.
\]
On $\{\sqrt{n\gamma_j}+Z_\parallel>0\}$, whose complement has
exponentially small probability,
\begin{align*}
  &\angle(\sqrt{n\gamma_j}\mathbf e_1+\mathbf Z,\mathbf e_1)
  >\theta_{n,M_j}
  \\
  &\hspace{15mm}\Longleftrightarrow\quad
  \|\mathbf Z_\perp\|
  >(\sqrt{n\gamma_j}+Z_\parallel)\tan\theta_{n,M_j}.
\end{align*}
The random threshold on the right therefore converges in probability to
$\sqrt n(\kappa_n\mu)^{-1/(n-1)}$.  Continuity of the
$\chi_{n-1}$ distribution
proves \eqref{eq:cone-error-curve}.

To compare the limiting curves, use
$D_n=(2/n)^a\kappa_n$ and define
\[
  \tau_\mu=\frac n2(\kappa_n\mu)^{-1/a}.
\]
Since $\mu D_n\tau_\mu^a=1$, both limiting curves can be expressed
through $T\sim\operatorname{Gamma}(a,1)$:
\begin{align*}
  F_n^{\mathrm{conv}}(\mu)
  &=\Prob\{T>\tau_\mu\},
  \\
  F_n^{\mathrm{ens}}(\mu)
  &=\E\left[1-
  \exp\left\{-\left(T/\tau_\mu\right)^a\right\}\right].
\end{align*}
Both curves are continuous and strictly increasing from zero to one.
Put $X=(T/\tau_\mu)^a$.  Its density is
\[
  f_X(x)=\frac{\tau_\mu^a}{a\Gamma(a)}
  e^{-\tau_\mu x^{1/a}},
  \qquad x>0,
\]
and is strictly decreasing.  Hence
\begin{align*}
  &F_n^{\mathrm{ens}}(\mu)-F_n^{\mathrm{conv}}(\mu)
  \\
  &\quad=\int_0^\infty
  \left(1-e^{-x}-\mathbf 1\{x>1\}\right)
  f_X(x)\,\dd x>0.
\end{align*}
Indeed, the unweighted positive and negative parts have the same mass,
\[
  \int_0^1(1-e^{-x})\,\dd x
  =\int_1^\infty e^{-x}\,\dd x=e^{-1},
\]
and strict decrease of $f_X$ gives
\[
  \int_0^1(1-e^{-x})f_X(x)\,\dd x
  >f_X(1)e^{-1}
  >\int_1^\infty e^{-x}f_X(x)\,\dd x.
\]
This proves \eqref{eq:limiting-error-curve-ordering}.

To obtain the quantitative threshold, let
\[
  \varphi_\gamma
  =\angle(\sqrt{n\gamma}\mathbf e_1+\mathbf Z,\mathbf e_1)
\]
and define its upper $\varepsilon$-quantile by
\[
  \Prob\{\varphi_\gamma>\vartheta_{n,\varepsilon}(\gamma)\}
  =\varepsilon.
\]
Put
$q_{n,\varepsilon}=\sqrt{s_{n-1,\varepsilon}}$.  Uniformly for $w$ in
a compact subset of $(0,\infty)$, conditioning on $Z_\parallel$ gives
\begin{align}
  &\Prob\left\{\varphi_\gamma>
  \arctan\left(\frac{w}{\sqrt{n\gamma}}\right)\right\}
  \nonumber\\
  &\quad=
  \E\left[
  Q_{n-1}\left(w^2(1+Z_\parallel/\sqrt{n\gamma})^2\right)
  \right]+O_n(e^{-n\gamma/2})
  \nonumber\\
  &\quad=Q_{n-1}(w^2)+O_n(\gamma^{-1}).
  \label{eq:quantitative-cone-angle-app}
\end{align}
The linear Taylor term vanishes because $\E[Z_\parallel]=0$; Gaussian
tail truncation justifies the uniform second-order expansion.  Since
\eqref{eq:uniform-cone-angle-expansion-app} keeps $w$ in a compact
subset of $(0,\infty)$, substituting
$w=\sqrt{n\gamma}\tan\theta_{n,M}$ in
\eqref{eq:quantitative-cone-angle-app} gives a uniform
$O_{n,K}(\gamma^{-1})$ approximation on every compact scaled-size
interval $K\subset(0,\infty)$.  Since the $\chi^2_{n-1}$ density is positive at
$s_{n-1,\varepsilon}$, quantile inversion yields
\begin{equation}
  \sqrt{n\gamma}\tan\vartheta_{n,\varepsilon}(\gamma)
  =q_{n,\varepsilon}+O_{n,\varepsilon}(\gamma^{-1}).
  \label{eq:cone-quantile-rate-app}
\end{equation}
Furthermore,
\[
  M_{\varepsilon,\mathrm{conv}}(n,\gamma)
  =\left\lfloor
  \Omega_n(\vartheta_{n,\varepsilon}(\gamma))^{-1}
  \right\rfloor.
\]
Combining \eqref{eq:cone-quantile-rate-app}, the small-angle expansion,
and the integer rounding error gives
\begin{align*}
  M_{\varepsilon,\mathrm{conv}}(n,\gamma)
  &=\mu_{n,\varepsilon}^{\mathrm{conv}}\gamma^a
  \left[1+O_{n,\varepsilon}(\gamma^{-1})\right]+O(1).
\end{align*}
Equivalently,
\[
  \frac{M_{\varepsilon,\mathrm{conv}}(n,\gamma)}{\gamma^a}
  =\mu_{n,\varepsilon}^{\mathrm{conv}}
  +O_{n,\varepsilon}
  \left(\gamma^{-\min\{1,a\}}\right).
\]
Taking logarithms, with the $O(1)$ rounding term included, gives
\eqref{eq:R-conv-expansion}.  Finally, strict
monotonicity of the two
limiting curves and \eqref{eq:limiting-error-curve-ordering} give
\eqref{eq:fixed-error-scale-ordering}.

\section{Rate-Bound Gap and Reliability--Blocklength Asymptotics}
\label{app:rate-gap-asymptotics}

This appendix proves Proposition~\ref{prop:high-snr-rate-gap} and
Theorem~\ref{thm:gap-reliability-blocklength}.

\subsection{Exact gap identity}

The constants in \eqref{eq:Dn-def} and \eqref{eq:kappa-n-def} satisfy
$D_n=(2/n)^a\kappa_n$.  Since
$s_{n-1,\varepsilon}=2t_{a,\varepsilon}$,
\[
  \mu_{n,\varepsilon}^{\mathrm{conv}}
  =
  \frac1{D_nt_{a,\varepsilon}^a}.
\]
The ensemble threshold equation likewise gives
\[
  \mu_{n,\varepsilon}^{\mathrm{ens}}
  =
  \frac{\mathcal L_a^{-1}(1-\varepsilon)}{D_n}.
\]
Thus the common factor $D_n^{-1}$ cancels exactly from
$r_{a,\varepsilon}$.
Using
$r_{a,\varepsilon}
=\mu_{n,\varepsilon}^{\mathrm{ens}}/
\mu_{n,\varepsilon}^{\mathrm{conv}}$ in
$F_n^{\mathrm{ens}}
(\mu_{n,\varepsilon}^{\mathrm{ens}})=\varepsilon$ gives
\[
  \mu_{n,\varepsilon}^{\mathrm{ens}}D_nT^a
  =
  r_{a,\varepsilon}
  (T/t_{a,\varepsilon})^a,
\]
which proves \eqref{eq:r-gap-root}.  As $r$ ranges from zero to
infinity, its left-hand side is continuous and strictly increasing from
zero to one, so the root is unique.
Equivalently, \eqref{eq:r-gap-root} can be written as
\[
  \mathcal L_a\!\left(
  \frac{r_{a,\varepsilon}}{t_{a,\varepsilon}^{a}}
  \right)
  =1-\varepsilon.
\]
Applying $\mathcal L_a^{-1}$ proves
\eqref{eq:r-gap-inverse-transform}.  Substitution into
\eqref{eq:gap-identity} proves \eqref{eq:gap-inverse-transform}.
The strict ordering \eqref{eq:fixed-error-scale-ordering} gives
$r_{a,\varepsilon}\in(0,1)$.  Equations
\eqref{eq:gap-identity} and \eqref{eq:finite-snr-gap-expansion} follow
from the two rate expansions.  Dividing the latter identity by
$R_{\varepsilon,\mathrm{conv}}(n,\gamma)$ and using
\eqref{eq:R-conv-expansion} proves
\eqref{eq:relative-shortfall-exact}.

\subsection{Fixed blocklength and increasing reliability}

Fix $a$ and write $t=t_{a,\varepsilon}$.  For
$X_t=(T/t)^a$,
\[
  \E[X_t^k]
  =
  \frac{\Gamma((k+1)a)}
       {\Gamma(a)t^{ak}},
  \qquad k\ge1.
\]
In particular,
\[
  \frac{\E[X_t^2]}{\E[X_t]^2}
  =\frac{\Gamma(3a)\Gamma(a)}{\Gamma(2a)^2}<\infty.
\]
The upper and lower quadratic Taylor bounds for $1-e^{-rX_t}$ and
inversion at the origin therefore give
\begin{equation}
  r_{a,\varepsilon}
  =
  \varepsilon
  \frac{\Gamma(a)t_{a,\varepsilon}^a}{\Gamma(2a)}
  \left[1+O_a(\varepsilon)\right].
  \label{eq:r-small-epsilon-app}
\end{equation}
The upper incomplete-gamma expansion gives
\[
  \varepsilon
  =
  \frac{t^{a-1}e^{-t}}{\Gamma(a)}
  \left[1+O_a(t^{-1})\right],
  \qquad t\to\infty.
\]
With $L=\ln(1/\varepsilon)$, inversion yields
\begin{equation}
  t_{a,\varepsilon}
  =
  L+(a-1)\ln L-\ln\Gamma(a)
  +O_a\left(\frac{\ln L}{L}\right).
  \label{eq:gamma-quantile-fixed-a-app}
\end{equation}
Substituting \eqref{eq:gamma-quantile-fixed-a-app} into
\eqref{eq:r-small-epsilon-app} and using
$G=-(n\ln2)^{-1}\ln r$ proves
\eqref{eq:gap-small-epsilon}.  Replacing $L$ by
$L+\ln10$ and subtracting proves
\eqref{eq:fixed-n-decade-cost}.

\subsection{Fixed reliability and increasing blocklength}

For fixed $0<\varepsilon<1$, a Berry--Esseen bound for
$T\sim\operatorname{Gamma}(a,1)$, followed by quantile inversion, gives
\[
  t_{a,\varepsilon}
  =
  a+z_\varepsilon\sqrt a+O_\varepsilon(1).
\]
Write $s=\ln r$ and make the exact change of variables
$y=a\ln(T/t_{a,\varepsilon})$ in \eqref{eq:r-gap-root}.  After
subtracting the gamma-tail probability, the root equation becomes
\[
  \int_{-\infty}^{\infty}d_s(y)w_a(y)\,\dd y=0,
\]
where
\begin{align*}
  d_s(y)
  &=1-e^{-e^{s+y}}-\mathbf 1\{y>0\},
  \\
  w_a(y)
  &=
  \frac{t_{a,\varepsilon}^a}{a\Gamma(a)}
  \exp\{y-t_{a,\varepsilon}e^{y/a}\}.
\end{align*}
For every fixed $C>0$, uniformly for $s$ in a fixed compact set and
$|y|\le C\ln a$,
\[
  \frac{w_a(y)}{w_a(0)}
  =
  1-\frac{z_\varepsilon y}{\sqrt a}
  +O_\varepsilon\left(\frac{1+y^2}{a}\right).
\]
The required kernel integrals are
\begin{align}
  \int_{-\infty}^{\infty}d_s(y)\,\dd y
  &=s+\gamma_{\mathrm E},
  \label{eq:boundary-I0-app}\\
  \int_{-\infty}^{\infty}y\,d_s(y)\,\dd y
  &=
  -\frac12(s+\gamma_{\mathrm E})^2-\frac{\pi^2}{12}.
  \label{eq:boundary-I1-app}
\end{align}
These identities follow by viewing $1-e^{-e^{s+y}}$ as the
distribution function of $\ln E-s$ for
$E\sim\operatorname{Exp}(1)$, whose mean and variance are
$-\gamma_{\mathrm E}-s$ and $\pi^2/6$.

The kernel $d_s(y)$ has integrable exponential tails, uniformly for $s$
in compact sets.  Dominated convergence first confines the root to a
fixed compact interval around $-\gamma_{\mathrm E}$; the central weight
expansion can then be integrated term by term, with the tails controlled
uniformly.  Using \eqref{eq:boundary-I0-app} and
\eqref{eq:boundary-I1-app} gives
\begin{align*}
  &\frac1{w_a(0)}
  \int_{-\infty}^{\infty}d_s(y)w_a(y)\,\dd y
  \\
  &\quad=
  s+\gamma_{\mathrm E}
  +\frac{z_\varepsilon}{\sqrt a}
  \left[
  \frac{(s+\gamma_{\mathrm E})^2}{2}
  +\frac{\pi^2}{12}
  \right]
  +O_\varepsilon(a^{-1}),
\end{align*}
uniformly on that interval.  Root monotonicity now yields
\[
  \ln r_{a,\varepsilon}
  =
  -\gamma_{\mathrm E}
  -\frac{\pi^2z_\varepsilon}{12\sqrt a}
  +O_\varepsilon(a^{-1}).
\]
Since $n=2a+1$ and
$G=-(n\ln2)^{-1}\ln r$, this proves
\eqref{eq:gap-large-n} and \eqref{eq:universal-nG-limit}.

\subsection{Exponentially stringent reliability}

Suppose \eqref{eq:joint-reliability-exponent} holds.  The gamma
large-deviation rate function is
\[
  \mathcal I(q)=q-1-\ln q,\qquad q>0.
\]
Consequently,
\[
  \frac{t_{a,\varepsilon_n}}a\longrightarrow q_\rho,
  \qquad \mathcal I(q_\rho)=2\rho.
\]

For $1<q_\rho<2$, make the change
$x=(T/t_{a,\varepsilon_n})^a$ in the root equation and normalize by the
transformed density $f_a$ at one.  Its normalized ratio satisfies
\[
  \frac{f_a(x)}{f_a(1)}
  =\exp\left\{-t_{a,\varepsilon_n}
  \bigl(x^{1/a}-1\bigr)\right\}
  \longrightarrow x^{-q_\rho}.
\]
Near zero, the root-equation kernel is $O(x)$, which is integrable
against $x^{-q_\rho}$ precisely when $q_\rho<2$; at infinity it decays
exponentially for $r$ in a compact neighborhood of the positive
limiting root.  Dominated convergence therefore yields
\begin{align*}
  0
  &=
  \int_0^\infty
  \left(
  1-e^{-rx}-\mathbf 1\{x>1\}
  \right)x^{-q_\rho}\,\dd x
  \\
  &=
  \frac{
  \Gamma(2-q_\rho)r^{q_\rho-1}-1
  }{q_\rho-1}.
\end{align*}
Monotonicity of the finite-$a$ and limiting roots therefore gives
\[
  r_{a,\varepsilon_n}
  \longrightarrow
  \Gamma(2-q_\rho)^{-1/(q_\rho-1)}.
\]
In particular, $G(n,\varepsilon_n)\to0$.

For $q_\rho\ge2$, use
\[
  (1-e^{-1})\min\{u,1\}
  \le1-e^{-u}\le\min\{u,1\},
  \qquad u\ge0.
\]
Define
\[
  J_a(r)
  =\E\left[\min\left\{1,r(T/t_{a,\varepsilon_n})^a\right\}\right].
\]
For a fixed trial $r=e^{-a\sigma}$ with $\sigma\ge0$, truncate $T/a$
to a compact interval, apply the gamma large-deviation principle there,
and then remove the truncation by exponential tightness.  This gives
\begin{align}
  -\frac1a\ln J_a(e^{-a\sigma})
  &\longrightarrow K_q(\sigma),
  \nonumber\\
  K_q(\sigma)
  &=
  \min\left\{
  1+\sigma+\ln q-2\ln2,\,
  \mathcal I(qe^\sigma)
  \right\},
  \label{eq:supercritical-Laplace-app}
\end{align}
where $q=q_\rho$.  With $(x)_+=\max\{x,0\}$, the variational form is
\[
  K_q(\sigma)
  =\inf_{u>0}
  \left\{\mathcal I(u)+\bigl(\sigma-\ln(u/q)\bigr)_+\right\}.
\]
Splitting the infimum at $u=qe^\sigma$ gives the two branches in
\eqref{eq:supercritical-Laplace-app}: the unsaturated branch is minimized
at $u=2$, whereas the saturated branch is minimized at its left
endpoint.  The unique solution of
$K_q(\sigma)=\mathcal I(q)$ is zero when $q=2$ and
\[
  \sigma=q-2-2\ln(q/2)
\]
when $q>2$.  At the exact root, the constant-factor bounds above imply
$-a^{-1}\ln J_a(r_{a,\varepsilon_n})\to \mathcal I(q)$.  When $q>2$, the
solution is positive.  Applying the Laplace limit at fixed trial values
on either side of this solution and using monotonicity shows that
$-a^{-1}\ln r_{a,\varepsilon_n}$ is squeezed between arbitrarily close
trial values on either side of the solution.

At the boundary $q=2$, the argument is one-sided because the solution
$\sigma=0$ lies at the edge of the domain $\sigma\ge0$.  The ordering
$r_{a,\varepsilon_n}<1$ gives
$-a^{-1}\ln r_{a,\varepsilon_n}\ge0$.  For every $\delta>0$,
$K_2(\delta)>\mathcal I(2)$, so the constant-factor bounds and the Laplace limit
show that the root-equation left-hand side at the trial value
$r=e^{-a\delta}$ is eventually smaller than $\varepsilon_n$.
Monotonicity then gives
$r_{a,\varepsilon_n}>e^{-a\delta}$ and hence
$\limsup[-a^{-1}\ln r_{a,\varepsilon_n}]\le\delta$.  Letting
$\delta\downarrow0$ proves the boundary case.  Thus
\[
  -\frac1a\ln r_{a,\varepsilon_n}
  \longrightarrow
  q_\rho-2-2\ln(q_\rho/2),
  \qquad q_\rho\ge2.
\]
Division by $n\ln2$ proves \eqref{eq:G-supercritical-reliability} for
$q_\rho>2$ and gives $G(n,\varepsilon_n)\to0$ at $q_\rho=2$.  This
boundary corresponds to
$\rho_{\mathrm c}=(1-\ln2)/2$, proving
\eqref{eq:critical-reliability-exponent}.

\section{Proof of the Prescribed-Gap Blocklength Result}
\label{app:blocklength-design-proofs}

This appendix proves
Proposition~\ref{prop:prescribed-gap-blocklength}.

For fixed $\varepsilon$, substitute
\[
  n=N_0(G_0)
  +\frac{c_0z_\varepsilon}{\gamma_{\mathrm E}}
  \sqrt{N_0(G_0)}+x
\]
into \eqref{eq:gap-large-n}.  Taking $x$ to be fixed and positive or
fixed and negative gives a two-sided localization.  More explicitly, the
$O_\varepsilon(n^{-2})$ remainder is uniform for $n$ in a fixed
relative neighborhood of $N_0(G_0)$.  Hence a sufficiently large
constant $K_\varepsilon$ satisfies the following, where
$b_\varepsilon(G_0)$ denotes the displayed center:
\begin{align*}
  n\le b_\varepsilon(G_0)-K_\varepsilon
  &\quad\Longrightarrow\quad G(n,\varepsilon)>G_0,\\
  n\ge b_\varepsilon(G_0)+K_\varepsilon
  &\quad\Longrightarrow\quad G(n,\varepsilon)<G_0,
\end{align*}
for all sufficiently small $G_0$, within that neighborhood.  The
leading $1/n$ term excludes smaller growing blocklengths, while every
fixed finite collection of blocklengths has a positive minimum gap and
is also excluded as $G_0\downarrow0$.  Thus the smallest qualifying
integer lies within $O_\varepsilon(1)$ of the displayed center.  This
proves
\eqref{eq:nG-small-gap-law}; subtracting the expansions at
$\varepsilon$ and $\varepsilon/10$ gives
\eqref{eq:nG-decade-law}.

For the ultra-reliable inversion, define
\[
  \mathfrak G(\rho)
  =
  \begin{cases}
  0,
  &0\le\rho\le\rho_{\mathrm c},\\[2pt]
  \dfrac{
  q_\rho-2-2\ln(q_\rho/2)
  }{2\ln2},
  &\rho>\rho_{\mathrm c}.
  \end{cases}
\]
This function is continuous and strictly increasing above
$\rho_{\mathrm c}$, and
$\mathfrak G(\rho_{G_0})=G_0$.  Put
$L=\ln(1/\varepsilon)$ and fix $0<\delta<\rho_{G_0}$.  The trial
blocklength
\[
  n_+(\varepsilon)
  =\left\lceil\frac{L}{\rho_{G_0}-\delta}\right\rceil
\]
satisfies $L/n_+(\varepsilon)\to\rho_{G_0}-\delta$ and is therefore
feasible for all sufficiently small $\varepsilon$.  Thus
$n_G(\varepsilon,G_0)\le n_+(\varepsilon)$ eventually.

For the matching lower bound, every integer
\[
  n\le
  \left\lfloor\frac{L}{\rho_{G_0}+\delta}\right\rfloor
\]
is eventually infeasible.  Otherwise, there would be a sequence
$\varepsilon_k\downarrow0$ and feasible $n_k$ in this range.  After
passing to a subsequence, one of three cases must hold: $n_k$ is
bounded; $n_k\to\infty$ and $L_k/n_k$ has a finite limit at least
$\rho_{G_0}+\delta$; or $L_k/n_k\to\infty$, where
$L_k=\ln(1/\varepsilon_k)$.  The first case is excluded by
\eqref{eq:gap-small-epsilon}; the second is excluded by
Theorem~\ref{thm:gap-reliability-blocklength} and the strict increase
of $\mathfrak G$ above $\rho_{\mathrm c}$.  It remains to exclude the
third case.

For a sequence in the third case, define locally in this proof
\[
  U_n=\frac{\Gamma(a)}{\Gamma(2a)}T^a,
  \qquad
  \E[U_n]=1,
\]
and let $\lambda_{n,\varepsilon}$ solve
$\E[1-e^{-\lambda_{n,\varepsilon}U_n}]=\varepsilon$.  Then
\[
  r_{a,\varepsilon}
  =
  \lambda_{n,\varepsilon}
  \frac{\Gamma(a)t_{a,\varepsilon}^a}{\Gamma(2a)}.
\]
Moreover,
\[
  \E[U_n^2]
  =
  \frac{\Gamma(3a)\Gamma(a)}{\Gamma(2a)^2}
  =
  \exp\{a\ln(27/16)+O(\ln a)\}.
\]
If $L_n=\ln(1/\varepsilon_n)$ and $L_n/n\to\infty$, the Taylor bounds
for $1-e^{-x}$ can be applied quantitatively.  Indeed,
$\varepsilon_n\E[U_n^2]\to0$ because $L_n/a\to\infty$; evaluating the
lower Taylor bound at $2\varepsilon_n$ and using $\E[U_n]=1$ gives
\[
  \varepsilon_n
  \le\lambda_{n,\varepsilon_n}
  \le2\varepsilon_n
\]
along every unbounded-blocklength subsequence.  A gamma-tail Chernoff
bound also gives
\[
  \ln\frac{t_{a,\varepsilon_n}}a
  =O\left(
  \ln\left(1+\frac{L_n}{a}\right)
  \right).
\]
Together with Stirling's formula and $L_n/a\to\infty$, this yields
\[
  \ln\left[
  \frac{\Gamma(a)t_{a,\varepsilon_n}^a}{\Gamma(2a)}
  \right]
  =o(L_n).
\]
Consequently, $-\ln r_{a,\varepsilon_n}\sim L_n$ and
$G(n,\varepsilon_n)\to\infty$, which is the required contradiction.
We have therefore bounded $n_G(\varepsilon,G_0)$ between
$L/(\rho_{G_0}+\delta)+O(1)$ and
$L/(\rho_{G_0}-\delta)+O(1)$.  Divide by $L$ and let
$\delta\downarrow0$ to prove
\eqref{eq:nG-ultrareliable-law}.

\bibliographystyle{IEEEtran}
\bibliography{bib-high-snr-spherical-code-paper-final}

\end{document}